\documentclass[11pt] {article}      
\usepackage{amsmath,amssymb, amsthm, eucal,accents}  
\usepackage{graphicx}
\usepackage{txfonts}
\usepackage[T1]{fontenc}     
\usepackage{enumerate}       
\usepackage{float}                 
\usepackage{mathtools}
\usepackage{tikz}
\usetikzlibrary{matrix, arrows, calc, fit}
\usetikzlibrary{arrows, decorations.markings}
\tikzstyle{AArrow} = [thick, decoration={markings,mark=at position 1 with {\arrow[semithick]{open triangle 60}}},%
   double distance=1.4pt, shorten >= 5.5pt,
   preaction = {decorate},
   postaction = {draw,line width=1.4pt, white,shorten >= 4.5pt}]
\tikzstyle{AArroww} = [semithick, white,line width=1.4pt, shorten >= 4.5pt]
\usepackage[strict]{changepage}
\usepackage{relsize}
\theoremstyle{plain}
\newtheorem{theorem}{Theorem}[section]            
\newtheorem{proposition}[theorem]{Proposition}  

\newtheorem{corollary}[theorem]{Corollary}	      %

\theoremstyle{definition}
\newtheorem{definition}[theorem]{Definition}

\numberwithin{theorem}{section}
\numberwithin{equation}{section}
\numberwithin{figure}{section}
\newcommand{\gaction}[2]{\genfrac{}{}{0.5pt}{}{#1}{#2}%
                        \!\lower2pt\hbox{\rotatebox[origin=c]{-90}{{$\looparrowright$}}}}

\newcommand{\dotfraction}[2]{\genfrac{}{}{0.5pt}{}{#1}{#2}%
                        \!\lower.5pt\hbox{{$\circ$}}}
\renewcommand\appendix{\par
  \setcounter{section}{0}
  \setcounter{subsection}{0}
  \setcounter{figure}{0}
  \setcounter{table}{0}
  \renewcommand\thesection{Appendix \Alph{section}}
  \renewcommand\thefigure{\Alph{section}\arabic{figure}}
  \renewcommand\thetable{\Alph{section}\arabic{table}}
}
\usepackage{titlesec}                                                           
\titlelabel{\thetitle.\ }                                                           
\titleformat*{\section}{\fontsize{14pt}{14pt} \bf}                   
\usepackage{fancyhdr}
\usepackage{sidecap}  
\usepackage[hang,small,sc]{caption}

\newcommand\myfootnote[1]{%
  \begingroup
  \renewcommand\thefootnote{}\footnote{#1}%
  \addtocounter{footnote}{-1}%
  \endgroup
}

\title{\bf Skein relations for spin networks, modified}

\author{Jerzy Kocik                  
\\ \small Department of Mathematics
\\ \small Southern Illinois University, Carbondale, IL62901
\\ \small jkocik{@}siu.edu  }
\date{}

\begin{document}

\maketitle

\begin{abstract}
\noindent
An alternative framework underlying connection between tensor ${\rm sl}_2$-calculus and spin networks is suggested.  
New sign convention for the inner product in the dual spinor space leads to a simpler 
and direct set of initial rules for the diagrammatic recoupling methods. 
Yet it preserves the standard chromatic graph evaluations. 
In contrast with the standard formulation,  the background space is that of symmetric tensor spaces, 
which seems to be in accordance with the representation theory of ${\rm SL}(2)$.  
An example of Apollonian disk packing is shown to be a source of spin networks.
The graph labeling is extended to non-integer values, resulting in the complex-values of chromatic evaluations. 
\\[3pt]
{\bf Keywords:}  Spin networks, skein relations, Apollonian disk configurations., ${\rm sl}(2)$, ${\rm su}(2)$.
\\

\noindent
{\bf MSC:} 
20C35, 
57M25, 
57M27 
.
\end{abstract}

\myfootnote{\hspace{-.27in} Appeared in {\it Journal of Knot Theory and Its Ramifications, doi.org/10.1142/S0218216518410031}}

%

\section{Introduction}

The graphical language for tensor calculus, like every language, comes in a number of dialects. 
One of its applications lies in spinor calculus and representation theory of ${\rm su}(2)$, 
and can be dated to the 50-ties of the last century \cite{BS, Lev,Lev1,Lev2,YLV,YB}. 
The next creative enhancement to this particular dialect is due to Roger Penrose together with his invention of spin networks \cite{Pe}.  
A connection with knot theory via recoupling mechanism was discovered by Louis Kauffman \cite{Ka}. 
The current revival of spin networks is due to their essential role in the loop quantum gravitation project
\cite{Ro}.  See also \cite{KLo}.
~ 

The essence of spin networks is the ``chromatic evaluation'' of certain labeled graphs, motivated by representation theory of ${\rm SU}(2)$.
A particular derivation of this graphical method was popularized in \cite{Pe, Ro, Ma, ba}. 
It however seems to have a number of ad hoc additional rules introduced to fix some sign problems. 
In Section 3 we show how to repair the problem by fixing the inner product in the dual spinor space.  
This alternative ``axiomatization'' seem to be more natural and pedagogically friendly.  
Despite different intermediate rules, the chromatic numerical evaluations coincide with the standard ones.

In Section 4, we introduce chromatic evaluations of circle arrangements, including Apollonian circle packings,
by transforming them into spin networks.  

Clearly, the formalism concerns both groups ${\rm SU}(2)$ and ${\rm SL}(2)$ and their Lie algebras, as they are coincide under complexification.

\section{Two dialects of diagrammatic language}

In this section we describe two dialects of diagrammatic language. 
The first,  {\bf arrow-tail dialect}, may be considered universal  
as it applies to any tensor system and does not introduce any shortcuts.  
The second, Penrose's special dialect, is popularly used for ${\rm SL}_2$ diagrammatics.  Here we present it as in \cite{Pe, Ma}.  
A new revised version is presented in Section \ref{sec:3}. 

~ 

\noindent { \bf A. Arrow-tail dialect.} 
This is the most robust yet the most flexible and universal dialect to which other may be translated.  
The idea is to be unambiguous with the possible price of loss of elegance. 
Tensors are represented by blocks, which may be labeled by the symbols of the tensors.  
The contravariant and the covariant entries are represented by arrows and tails, respectively.  Here are examples:  

\smallskip
\includegraphics[scale=.9]{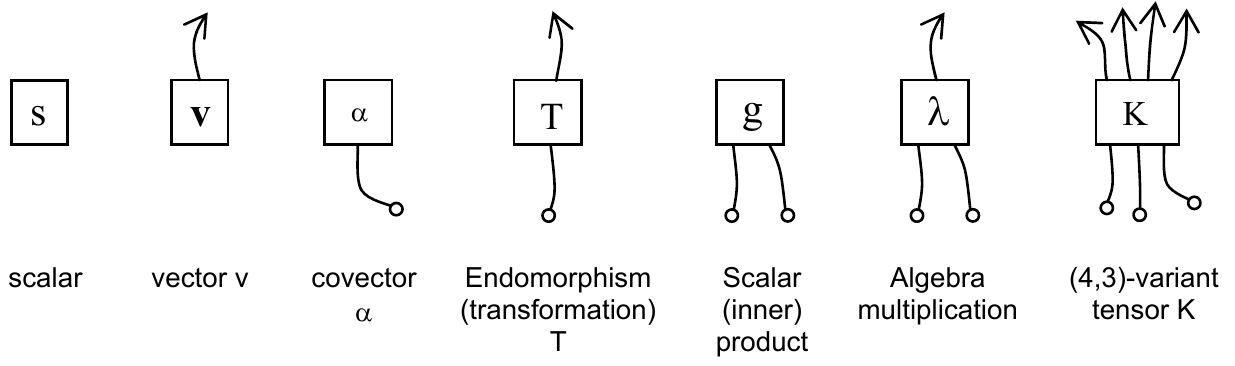}

\noindent Na\"ively, one may think of the arrows and tails as indices in the basis description.  
An arrow represents a contravariant (upper) index while the tail the covariant (lower index). 
The shape and the place of attachment of the arrows and tails is inessential, for instance:

\smallskip
\includegraphics[scale=.9]{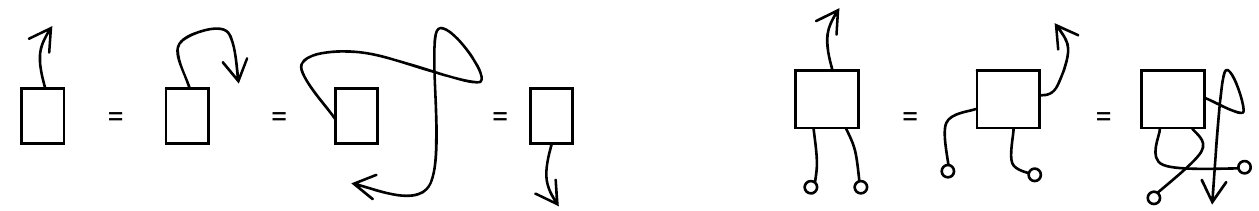}

\noindent However the order the arrows exit the block matters:  
the contravariant arrows are oriented clockwise and the covariant counterclockwise.  
If you want to change the order of two (or more) indices, the crossing like in the last figure is not sufficient in this convention.  
One needs to do it by applying an appropriate tensor.  
For simplicity, we shall denote it by a rectangle or circle with the crossing as shown:

\smallskip
\includegraphics[scale=.9]{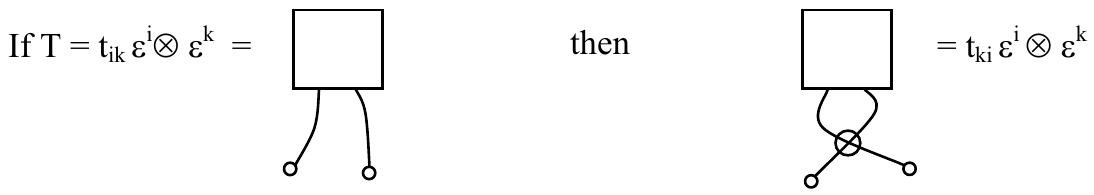}

\noindent 
In general, if a tensor product of a number of tensors with possible mutual contractions is considered
as a single entity,  possibly with rearranged order of indices, 
then either the trick described above may be used or a new encompassing frame will indicate the new situation.  
Here are the two conventions:

\begin{center}
\includegraphics[scale=.9]{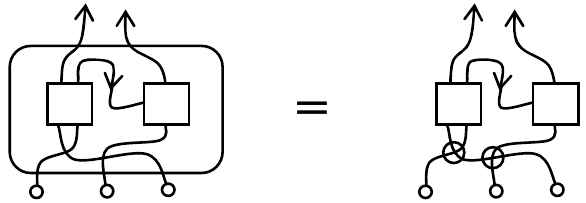}
\end{center}

\noindent 
Various tensor terms and equations may be now represented diagrammatically by joining arrows with tails 
(corresponding to summing up along the corresponding indices).  
Here are a few examples.  
Note that the arrow and the tail under contraction may be drawn by a single path with an arrow head on it:

\smallskip
\begin{center}
\includegraphics[scale=.9]{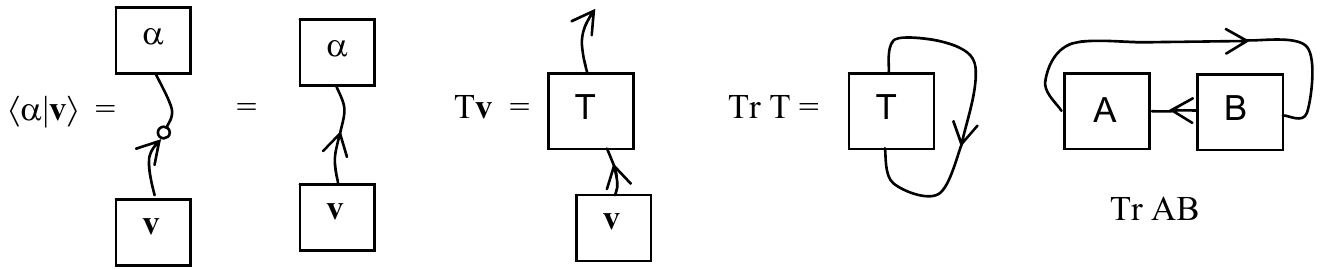}
\end{center}

\noindent 
Any inner product $g: V\otimes V \;\rightarrow\; \mathbb F$ (not necessarily symmetric or non-degenerate) is equivalent 
to a map $\dot{g}: V\;\rightarrow \;V^*$  from the space to the dual space (the dot atop will be dropped if no confusion arises). 
Both maps are ``induced maps'' from a single tensor $g$.

\begin{equation}  
\label{eq:inner}
\includegraphics[scale=.9]{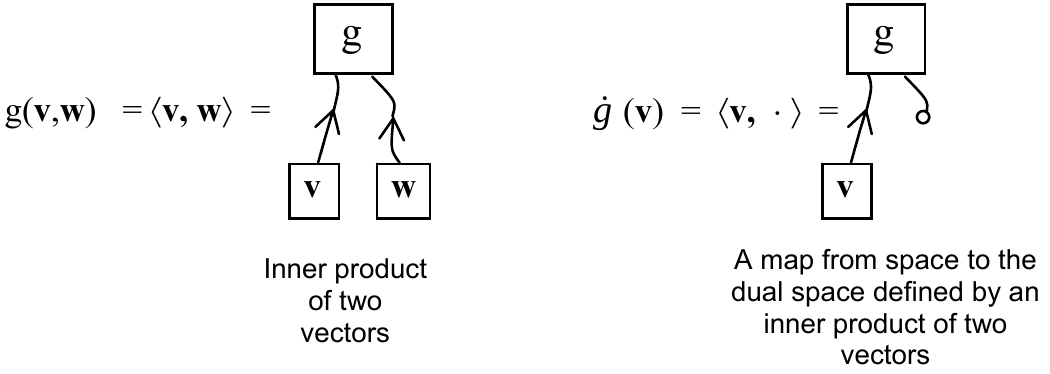}
\end{equation}

\noindent 
{ \bf B.  Penrose binor calculus (spin-network dialect)} (Penrose's ``loop notation'').  
This is a specification, proposed by Penrose and having its roots in some older primitive versions designed 
for Clebsch-Gordan recoupling theory for irreducible representations of ${\rm SU}(2)$.  
One starts with a 2-dimensional space $V$.  
Since there are generally two essential central tensors used, a skew-symmetric 2-covariant tensor $\epsilon$ 
(denoted in physics literature $\epsilon_{AB}$) and  a skew-symmetric 2-contravariant tensor 
(denoted $\epsilon^{AB}$ ), which he represents in some convenient (symplectic) basis by the following matrices
\begin{equation}  
\label{eq:Penrose}
       \varepsilon _{AB} 
                 \quad \leftrightarrow \quad 
       \left[\begin{array}{cc} {0} & {1} \\ {-1} & {0} \end{array}\right]
\qquad\hbox{and}\qquad
       \varepsilon ^{AB} 
                   \quad \leftrightarrow \quad 
       \left[\begin{array}{cc} {0} & {1} \\ {-1} & {0} \end{array}\right]
\end{equation}

\noindent 
The Penrose's diagrammatic notation may be considered as a specification of the universal dialect:  (1) 
the page is given a chosen orientation, upwards; (2) the box notation for the structural `epsilons'' is simplified to arcs:
\begin{equation}  
\label{eq:2.1}
\includegraphics[scale=1]{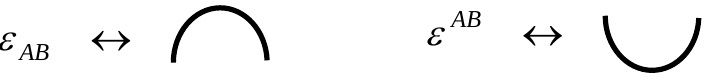}
\end{equation}  
respectively. Now, one easily checks that the following identities follow
\begin{center}
\includegraphics[scale=.9]{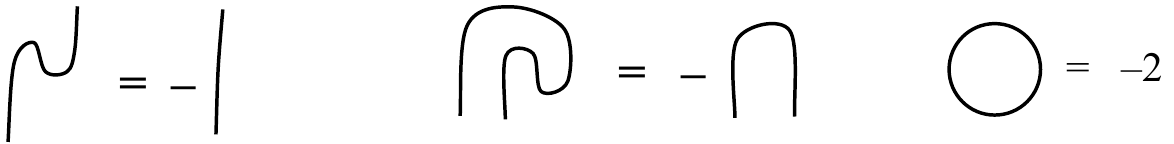}
\end{center}

\noindent 
Terms which involve the various contractions of the above two tensors may be interpreted as 2D ``links'' 
with the following convention:  

\begin{enumerate}
\item[(a)] vertical direction is selected, 
\item[(b)]  minus sign is attributed to every minimum (concave part of the curve), 
\item[(c)] a minus is associated to every crossing. 
\end{enumerate}


\noindent 
We will show how this can be fixed by simple rectification of the sign in structure tensor $\varepsilon$ of the inner product 
in the dual space. 
We resulting rules seem more natural and simpler. This will be shown in the following section.
Yet, the consequent values of evaluations are the same.

\section{ Revised spinor sl${}_{2}$-calculus}
\label{sec:3}

Let $V$ be a 2-dimensional vector space over some field $\mathbb F$, equipped with a skew-symmetric (symplectic) product $\omega$.  
We shall use alternative notations:
\begin{equation}
\label{eq;3.0}
                  \langle \mathbf v, \mathbf w\rangle = \omega (\mathbf v,\mathbf w)
\end{equation}
One may chose a convenient basis $\{\mathbf e_{1},\, \mathbf e_{2}\}$  
in which $\omega(\mathbf e_{1},\, \mathbf e_{2}) = 1 = -\omega(\mathbf e_2,\, \mathbf e_{1})$.  
Relating to \eqref{eq:inner}, we view the inner product as a homomorphism $\dot{\omega }: V\rightarrow V^*$.  
The responsible tensor (bi-form) will be denoted simply by $\omega$.  

\smallskip

It is the inverse map   $\dot{\omega }^{-1} :V^* \rightarrow V$ 
which should define the inner product on the dual space $V^*$,
namely via
\begin{equation}  
\label{eq:3.1}
           \Omega(\alpha,\beta) \ = \   \langle \beta\;|\; \dot{\omega }^{-1} (\alpha)\rangle
\end{equation} 
where $\langle \,\cdot \,| \,\cdot\,\rangle$ denotes the natural pairing of covectors and vectors 
(not to be confused with the inner product $\langle\,\cdot \,, \,\cdot\,\rangle$ ).   
In the universal diagrammatic language we have

\includegraphics[scale=.89]{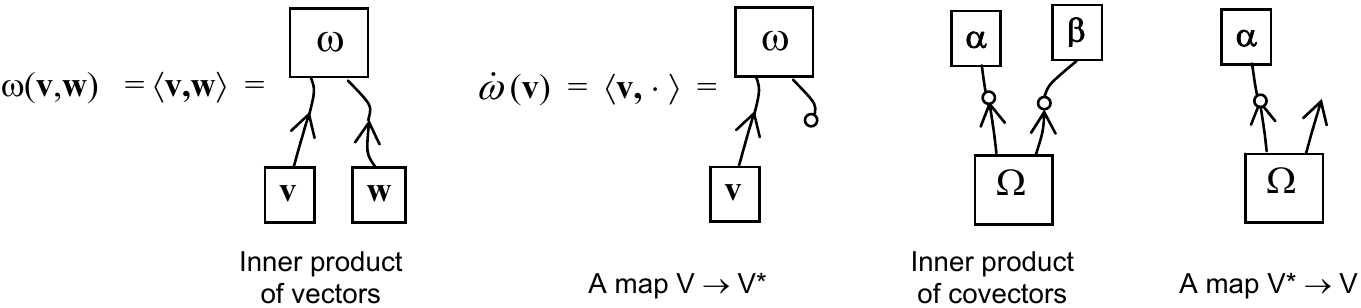}

\noindent
where $\Omega$ is the bi-contravariant tensor underlying the inverse map $\dot{\omega }^{-1} $.  
In the matrix notation, the conclusion is a different association of matrices than the standard \eqref{eq:Penrose}:
\begin{equation}  
\label{eq:3.2}
\omega =\left[\begin{array}{cc} {0} & {1} \\ {-1} & {0} \end{array}\right]    
\qquad      
\Omega =\left[\begin{array}{cc} {0} & {-1} \\ {1} & {0} \end{array}\right]
\end{equation}  
Now, like in the Penrose's dialect, we orient the page upward as in his setup, but 
interpret the half-arcs as shown:
\begin{equation}  
\label{eq:3.3}
\includegraphics[scale=1]{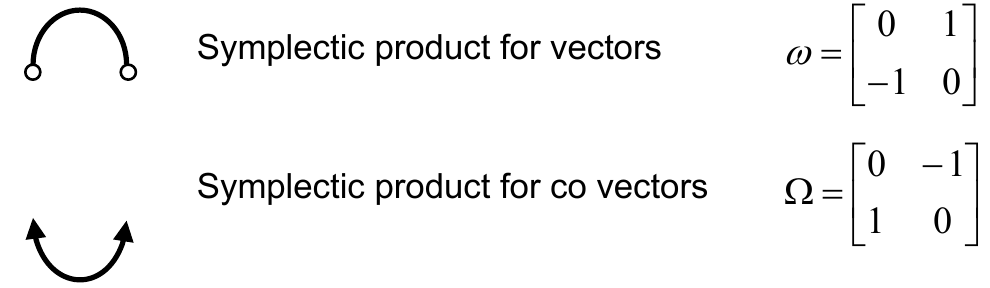}
\end{equation}  

In the following, the basic properties are derived:

\begin{proposition}
\label{thm:3.1}
The following wavy diagram may be straightened:

\includegraphics[scale=1]{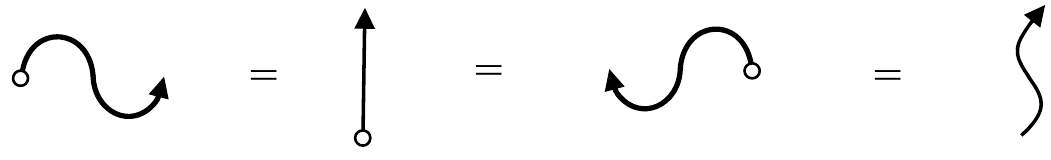}
\end{proposition}

\begin{proof}   
$\omega\Omega = \Omega\omega = I$  (identity matrix).  
The reader may want to check by acting on $\mathbf e_1$ and $\mathbf e_2$ 
to get convinced that the matrices correctly represent the situation. 
\end{proof}

\begin{corollary}
\label{thm:3.222}
Any number of indentations can be removed without change of sign.
\end{corollary}

\noindent Examples:
$$
\includegraphics[scale=1]{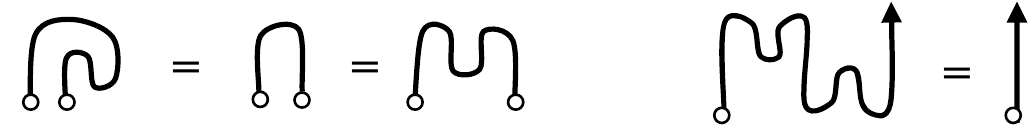} 
$$
Note that the line assuming a horizontal direction part does not cause ambiguity:  
One may bend the horizontal part up or down as both choices lead to the same simplification.  
$$
\includegraphics[scale=1]{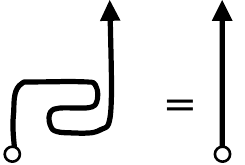} 
$$
\begin{proposition}
\label{thm:3.2}
Circle and self-crossed circle have the following numerical values
\begin{equation}  
\label{eq:3.4}
\includegraphics[scale=1]{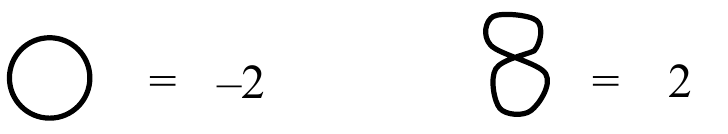} 
\end{equation}  
\end{proposition}

\begin{proof}
This is the simple result of taking a trace+ 
\begin{equation}  \label{eq:}
\includegraphics[scale=.91]{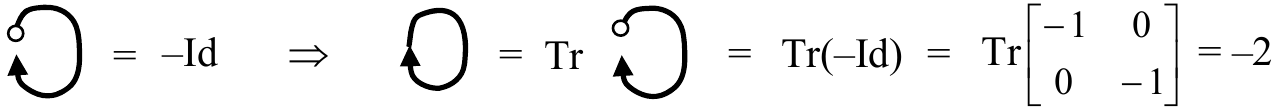} 
\end{equation}  
The case of the 8-figure goes in the same way.
\end{proof}     

%

\begin{proposition}[Skein relation for $sl_2$ calculus]
\label{thm:3.3}
 The diagram with a crossing may be replaced by a formal sum of two diagrams where 
the crossing is replaced by each of  the two possible non-crossed connections :

\begin{equation}  \label{eq:3.5}
\includegraphics[scale=1]{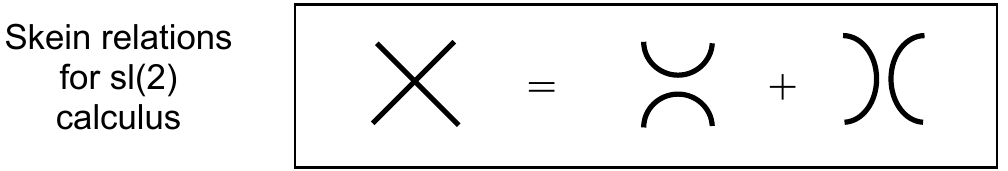} 
\end{equation}  

\end{proposition}

\begin{proof}   
The following skein relation is implied directly:
\begin{equation}  \label{eq:3.6}
\includegraphics[scale=.91]{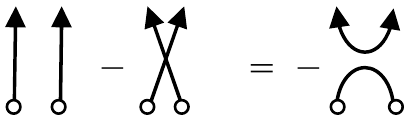} 
\end{equation}  
Indeed, suppose we apply the left-hand side to a pair of vectors $\mathbf v$ and $\mathbf w$, or equivalently,
to $\mathbf v\otimes\mathbf w$, we get  $\mathbf v \otimes \mathbf w - \mathbf w \otimes \mathbf v$.
 Express the vectors in the symplectic basis $\{\mathbf e_1,\mathbf e_2\}$ such that 
$\omega(\mathbf e_1,\mathbf e_2) = 1 = \omega(\mathbf e_2,\mathbf e_1)$  (see \eqref{eq:3.2}),
$$  
\mathbf v = v^1 \mathbf e_1 + v^2\mathbf e_2 
\quad \hbox{and}\quad   
\mathbf w = w^1 \mathbf e_1 + w^2 \mathbf e_2\,,
$$  
and simplify:
\begin{equation}  
\label{eq:liczenie}
\begin{array}{lll}
&(v^1\mathbf e_1 + v^2 \mathbf e_2) \otimes (w^1 \mathbf e_1 + w^2 \mathbf e_2)  
     -  (w^1 \mathbf e_1 + w^2 \mathbf e_2) \otimes (v^1 \mathbf e_1 + v^2 \mathbf e_2)\\
    &\hspace{2in} = (v^1w^2 - v^2w^1) (\mathbf e_1 \otimes \mathbf e_2 - \mathbf e_2 \otimes \mathbf e_1 )   \\
    &\hspace{2in} = \omega(\mathbf v,\mathbf w) \; (-\Omega )    =  -  \omega(\mathbf v,\mathbf w) \cdot \Omega 
\end{array}
\end{equation}

\noindent 

\noindent Hence, removing the information about the vectors, we see that the left hand side corresponds to 
a tensor $-\omega\otimes\Omega$, illustrated on the right side of \eqref{eq:3.6}.  
Reorganizing \eqref{eq:3.6} leads to the main statement.    
\end{proof}

The above skein relation differs from the standard diagrammatic language for ${\rm SL}_2$
by the signs.  
The summary of all differences are given in Table 1 at the end of Section \ref{sec:4}.

Thus we have the following rules:  
A { \bf simple ${\rm sl}_2$--diagram} consists of a number of closed loops possibly self-intersecting and intersecting each other.  
A { \bf closed ${\rm sl}_2$--diagram} is a formal linear combination of simple ${\rm sl}_2$-diagrams. 
There is an evaluation which associates to a diagram a number.  
It can be calculated by using the skein relation to turn it into a combination of collections 
of nonintersecting circles and then replacing the circles by the value (-2).  
To be clear:
$$
\includegraphics[scale=.91]{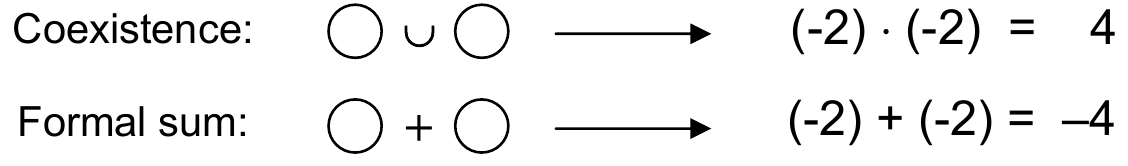} 
$$
\noindent { \bf Remark:}  There is a similarity to Kauffman's bracket for links \cite{Ka}.  Recall that
$$
\includegraphics[scale=1]{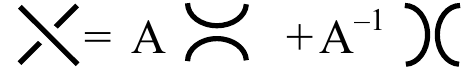} 
$$
\noindent 
Substituting $A = 1$ and disregarding the information 
on the vertical position of the lines at the crossing (above versus below), we get skein relation
\eqref{eq:3.5}.  The standard convetion correspond to $A=-1$ (See Table 1, Sec 4).
Interestingly, the evaluation of a circles is in both cases $d = -(A^2 + A^{-2})=-2$.  
We may call this disregard of crossing a ``diagrammatization'' of the link/knot.

\begin{proposition}
\label{thm:twist}
A single twist changes the sign of a path
$$
\includegraphics[scale=1]{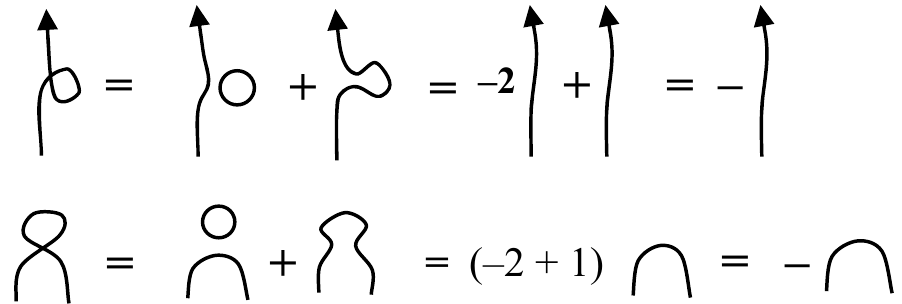} 
$$
\end{proposition}

We can have a new look at Proposition \ref{thm:3.2}.  
The twisted version of the circle results by changing the sign of the value of the untwisted circle.  
Alternatively, one may use the skein identity and get the same result:
$$
\includegraphics[scale=1]{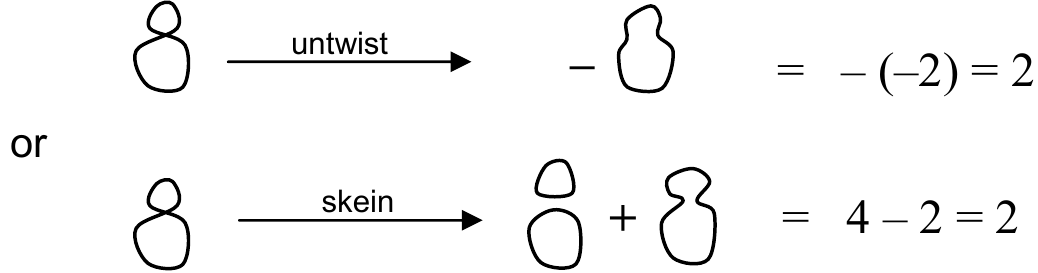} 
$$

\noindent 
{\bf Remark on the band interpretation of the change of sign for twist:}  
Replace the path by a paper strip; 
say vertically extending above the page.  Now, if you straighten out a twist-free path the result will be an unbend strip.  
But if you do the same to a line with a single twist, the strip will be rotationally bent by $360^\circ$.  
And if you do the same to a line with two twists, the resulting strip will be unbent or possibly bent by $720^\circ$.  
Hence we have the same phenomenon as in the case of the double degeneracy of rotations, related to double covering of ${\rm SO}(3)$
by ${\rm SU}(2)$  (or ${\rm SO}(1,2)$ by ${\rm SL}(2,\mathbb R)$). 
$$
\includegraphics[scale=1]{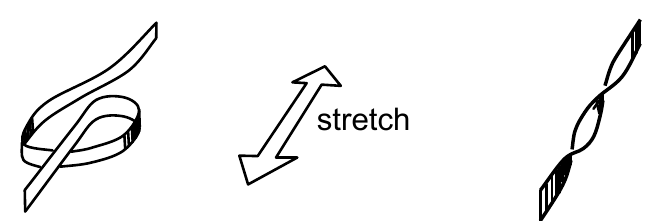} 
$$
\begin{proposition}
\label{thm:reidem}
Reidemeister moves and the corresponding identities R1, R2 and R3 for diagrammatic ${\rm sl}_2$ calculus are:
$$
\includegraphics[scale=1]{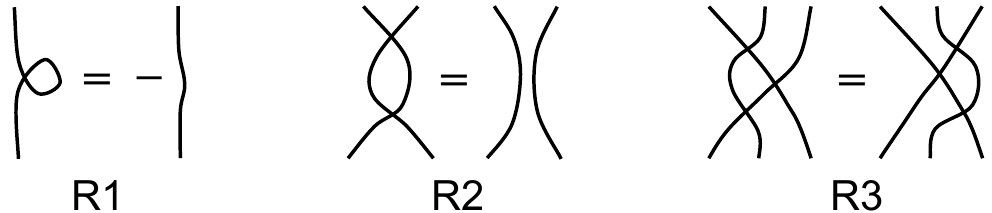} 
$$
\end{proposition}

\begin{proof}   
The first identity is proved as Proposition \ref{thm:twist}. 
Here is the proof of R2.  
The heavy dots indicate which of the intersections are about to be resolved in the next step (by the skein relation):
$$
\includegraphics[scale=.9]{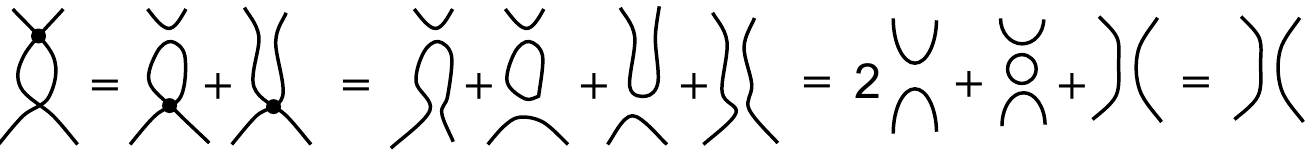} 
$$
\noindent where evaluation  $\langle \rm o\rangle = -2$ was used at the end.  Proof of R3:
$$
\includegraphics[scale=.9]{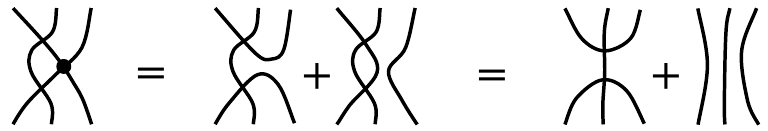} 
$$
(Use R2 in the second step).  
This ends the proof, since the right hand side, by symmetry, resolves to the same sum. 
\end{proof}   

\noindent 
{\bf Remark:} Show that the right hand side of the last diagram may be resolved further 
into the following sum of non-intersecting paths:
$$
\includegraphics[scale=.9]{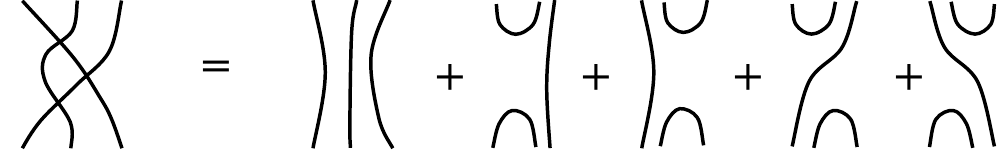} 
$$
\noindent Superfluous points of intersection like the ones below do not affect the valuations:
$$
\includegraphics[scale=.9]{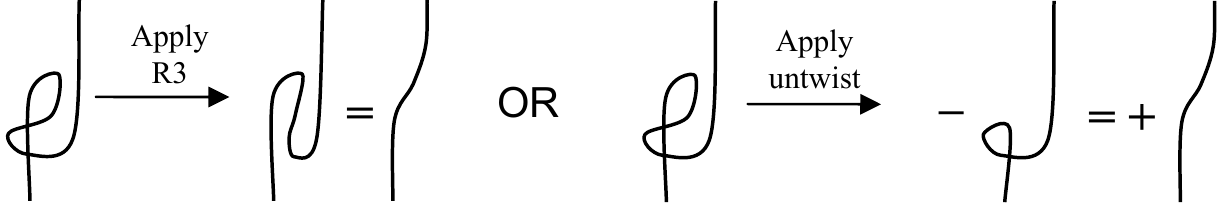} 
$$

\begin{definition}
\label{def:3.6}
A symmetrizing tensor $S$ of degree $n$ is an $(n,n)$-variant tensor totally symmetric in the contravariant sector 
and covariant sector, which, acting on $n$  vectors $\{\mathbf v_1,\dots , \mathbf v_n\}$ (or, equivalently, their tensor product) gives: .
\begin{equation}  
\label{eq:3.7}
S(\mathbf v_{1} \otimes ...\otimes \mathbf v_{n} )  
        =       \frac{1}{n\, !} \sum _{\sigma }\mathbf v_{\sigma (1)}  \otimes ...\otimes \mathbf v_{\sigma (n)} 
\end{equation}  
where the sum extends over all elements $\sigma$ of the symmetry group ${\rm S}_n$.  For example
$$
\begin{array}{ll}
S(\mathbf v\otimes \mathbf w)  
     &=   \frac{1}{2} ( \mathbf v \otimes \mathbf w + \mathbf w \otimes\mathbf v )\\[7pt]
S(\mathbf v\otimes \mathbf w  \otimes \mathbf u) 
  & =  \frac{1}{6} (
   \mathbf v \otimes \mathbf w \otimes \mathbf u 
+   \mathbf v \otimes \mathbf u \otimes \mathbf w 
+   \mathbf w \otimes \mathbf v \otimes \mathbf u \\
& \quad \ +   \mathbf w \otimes \mathbf u \otimes \mathbf v 
+   \mathbf u \otimes \mathbf v \otimes \mathbf w 
+   \mathbf u \otimes \mathbf w \otimes \mathbf v )
\end{array}
$$
\end{definition}

\noindent 
In the diagrammatic notation it will be marked by a transversal line.  
Here are simple examples: 
$$
\includegraphics[scale=.9]{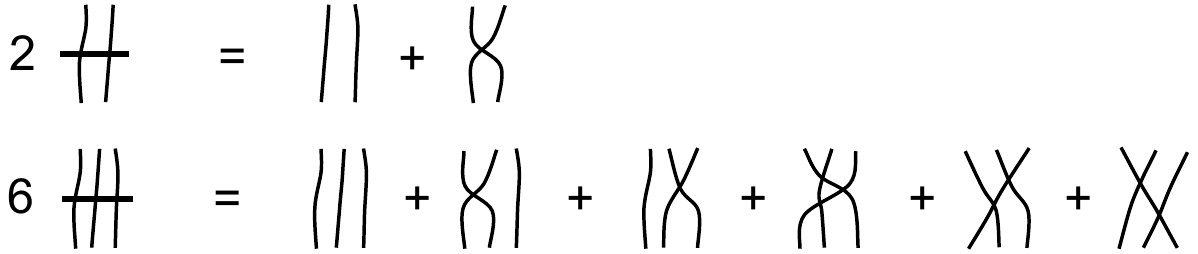} 
$$
(For transparency, the numerical factor is located on the left side.)  
Note that this operator is a projection, $S^2 = S$, or,   diagrammatically:
$$
\includegraphics[scale=.9]{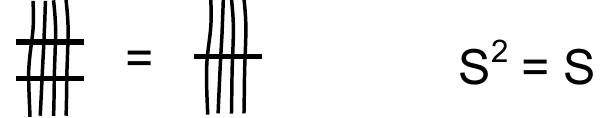} 
$$

\noindent 
Convenient identities: 
$$
\includegraphics[scale=.9]{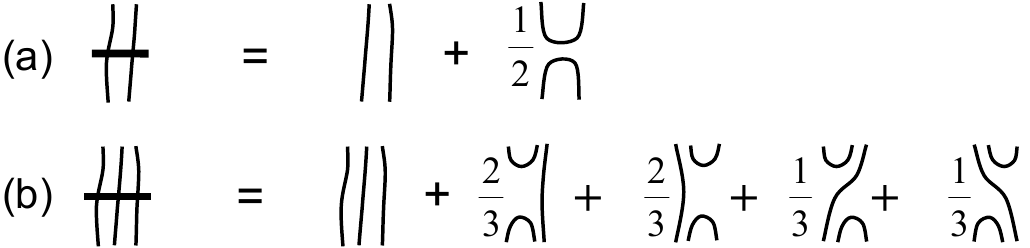} 
$$
{\bf Notational convention:} 
A bundle of lines with the symmetrizing tensor applied to it will be represented by a single line with the number $n$ on it or next to it.  
For example:
\begin{equation} 
\label{eq:3.8}
\includegraphics[scale=.9]{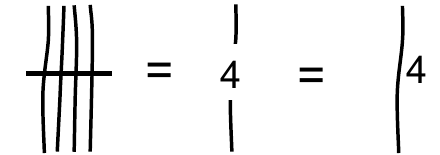} 
\end{equation}

\section{Spin networks}
\label{sec:4}

{ \bf Spin network} (or simply a { \bf net})
is a three-valent graph (vertices and edges) with the edges labeled by natural numbers, say $a$, $b$ and $c$.  
The sum of the numbers of the edges adjacent to any vertex must be even and any of them must not exceed
the sum of the remaining two:   $a,b,c \leq (a+b +c)/2$. 

~

\noindent Every spin network G may be given evaluation $\langle G \rangle$  in these three steps:

\begin{enumerate}
\item  Replace every edge with label $k$ by a group of $k$ nonintersecting {\bf strands} 
with a symmetrizing tensor, like one shown in \eqref{eq:3.8}

\item  Replace every vertex by a nonintersecting connection of the meeting strands

\item  Interpret it as a formal sum of ${\rm sl}_{2}$ diagrams, sum up their evaluations, 
and average by dividing by the number of the diagrams.  
\end{enumerate}

\noindent For example this 3-valent vertex may be viewed as follows:
$$
\includegraphics[scale=.9]{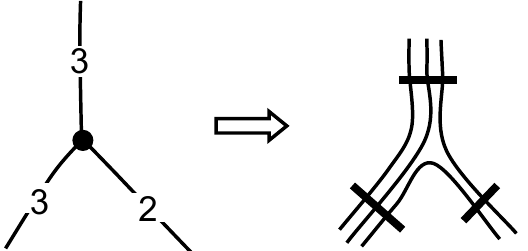} 
$$
The map $G \;\rightarrow\; \langle G \rangle$  is called the { \bf chromatic evaluation} of $G$ 
(suggesting the labels to be viewed as ``colors'').  Here are two simple examples

\

\noindent 
{\bf Example 1:}
$$
\includegraphics[scale=.9]{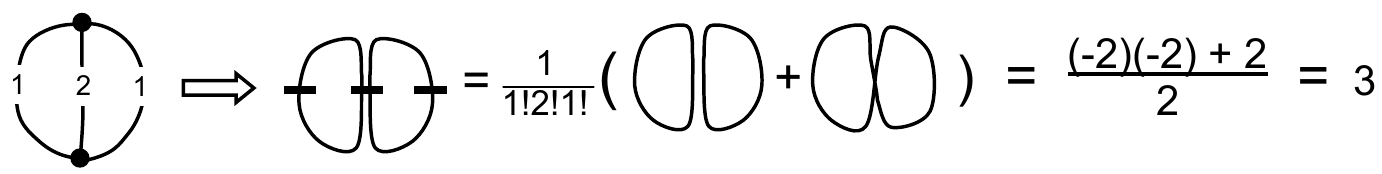} 
$$

\noindent 
{\bf Example 2}:
$$
\includegraphics[scale=.8]{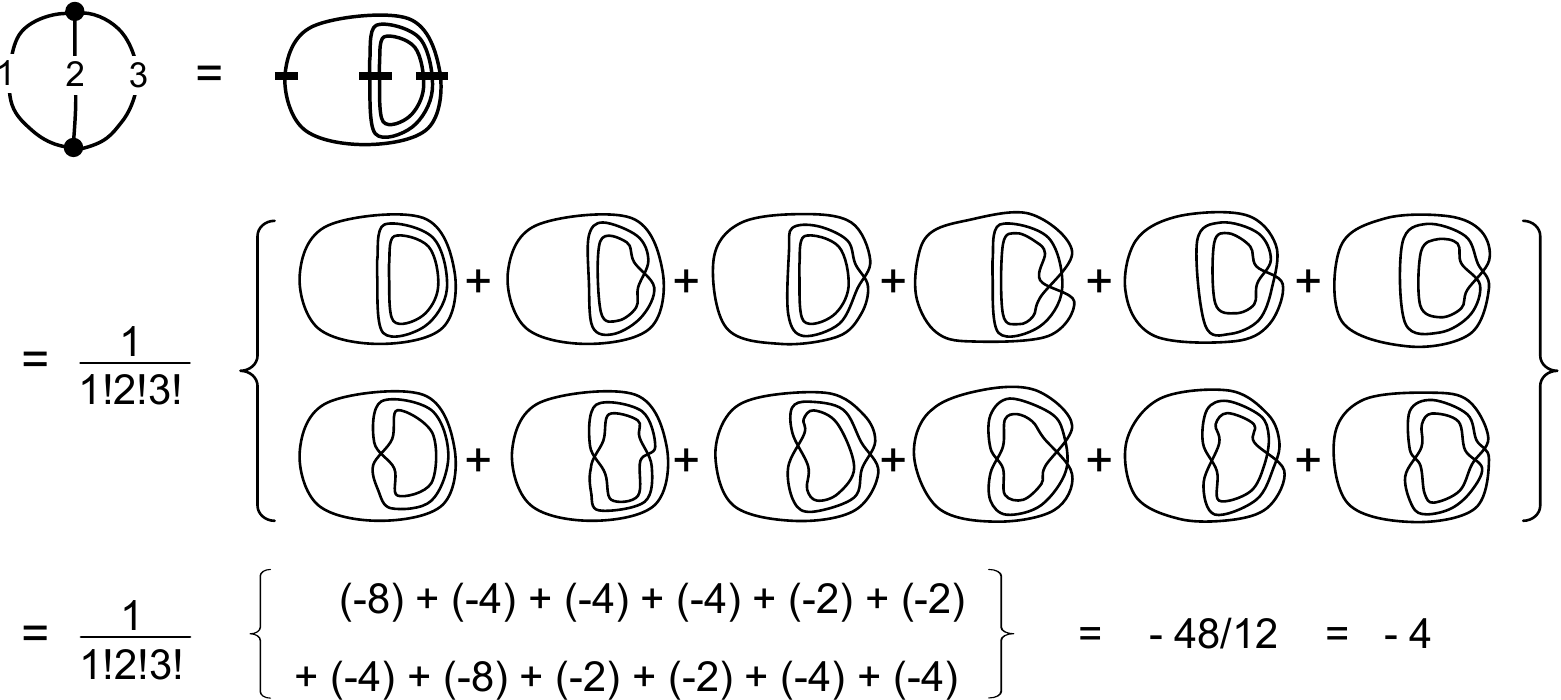} 
$$

\noindent Each set of loops due to particular choice of the permutations of the strands will be called a { \bf state}. 
 All states form a { \bf resolution}. The resolution in the first example consists of two states, in the second of 12.  
Each state consists of a number of (possibly self-intersecting) circles. 
The procedure of chromatic evaluation may then be summarized symbolically as follows:
$$
\includegraphics[scale=.8]{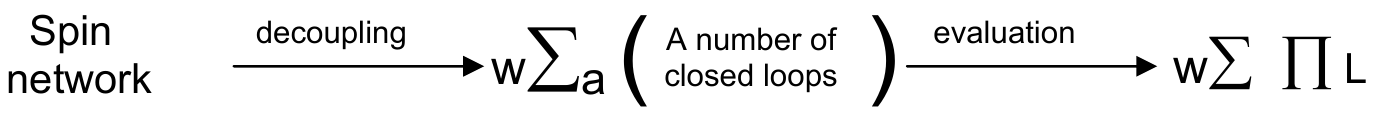} 
$$

\noindent
where in the middle the sigma represents a formal sum over all states of the net and $w$ is the weight coefficient, 
the reciprocal of the number of the resolutions.  
In the third step, the sum becomes arithmetic and each resolution is replaced by 
$$
           L = (-1)^{\hbox{\small\rm(nr of even loops}}\; 2^{\hbox{\small\rm(nr of loops)}}\,.
$$
The set of states forms a Cartesian product  $\mathbb Z_a\times \mathbb Z_b \times \dots \times \mathbb Z_c$, 
where $\{a,b,\dots ,c\}$ is the collection of the labels.

\bigskip

\noindent 
The main observation:  The evaluations obtained via Penrose-inspired rules and developed in \cite{KL}, 
coincide with evaluations presented in the present paper.

\noindent Here we recall the fundamental three results of \cite{KL} as they are valid to our method as well.  
They were obtained with the help of the  Temperley-Lieb algebra and are related to the Clebsch-Gordan 
symbols of the representation theory of ${\rm SL}(2)$.  

\begin{proposition}
\label{thm:4.1}
Formulas of simple spin networks:
%
%
%
\begin{equation} 
\label{eq:4.1}
\begin{array}{clc}
\raisebox{-20pt}{\includegraphics[scale=.9]{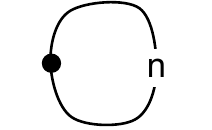}} & 
    = \Delta (n)=(-1)^{n} (n+1)     
    &\quad (a)\\[25pt]
\raisebox{-20pt}{\includegraphics[scale=.9]{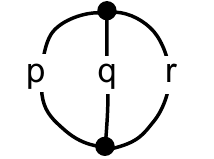}} & 
    =  \theta (p,q,r)=(-1)^{i+j+k} \dfrac{i!\,j!\,k!\,(i+j+k+1)!}{(i+j)!\,(j+k)!\,(k+i)!}      
    &\quad (b)\\[-5pt]
&\qquad \hbox{\rm where} \ \  p=i+j,\; \; q=j+k,\; \; r=k+i\\[7pt] 
\raisebox{-20pt}{\includegraphics[scale=.9]{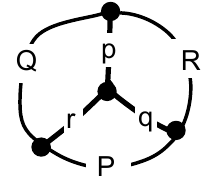}} & 
=    Tet\left[\begin{array}{ccc} {P} & {Q} & {R} \\ {p} & {q} & {r} \end{array}\right]
           =   \dfrac{\prod _{i,j}(b_{i} -a_{j} )! }{p\, !q!r!P!Q!R!} \; 
                       \!\!\! \displaystyle\sum _{s=\max \{ a_{j} \} }^{\min \{ b_{i} \} } \dfrac{(-1)^{s} (s+1)!}{\prod _{i}(s-a_{i} )!\prod _{j}(b_{j} -s)!  }      
     &\quad (c)
\end{array}
\end{equation}
where the terms are defined as follows:
$$
\begin{array}{llll}
&a_{1} =(p+q+r)/2   &\qquad\quad&b_{1} =(p+P+q+Q)/2\\
&a_{2} =(P+Q+r)/2   &&b_{2} =(p+P+r+R)/2\\  
&a_{3} =(P+q+R)/2   &&b_{3} =(q+Q+r+R)/2\\  
&a_{4} =(p+Q+R)/2
\end{array}
$$ 
or simply:
$$
\begin{array}{lll}
&a_i \ = \ \hbox{\rm sum over edges adjusted to vertex \textit{i} (four cases)}\\
&b_j \ = \ \hbox{\rm sum over one of the closed quadrangle through 4 vertices (3 cases)}
\end{array}
$$
\end{proposition}

\noindent 
In the description of Tet, the labeling is changed with respect to the original in \cite{KL} to emphasize the natural duality of the opposite edges, 
$p\leftrightarrow P$, 
$q \leftrightarrow Q$, 
$r \leftrightarrow R$, 
thanks to which the symmetries in the formulas become more transparent.

\begin{corollary}
\label{thm:4.3}
The chromatic functions satisfy the following symmetries:
\begin{enumerate}
\item
Symmetries of the theta function:  any permutation of the three terms.\\[-5pt]
\item
Symmetries of Tet:  any permutation of the columns and any vertical flip of two columns simultaneously. 
Here are 4 of 24 terms:
\begin{equation}  \label{eq:}
   Tet\left[\begin{array}{ccc} {A} & {B} & {C} \\ {a} & {b} & {c} \end{array}\right]
 =Tet\left[\begin{array}{ccc} {a} & {b} & {C} \\ {A} & {B} & {c} \end{array}\right]
 =Tet\left[\begin{array}{ccc} {A} & {b} & {c} \\ {a} & {B} & {C} \end{array}\right]
  =Tet\left[\begin{array}{ccc} {B} & {A} & {C} \\ {b} & {a} & {c} \end{array}\right]= ... etc
\end{equation} 
\end{enumerate}

\end{corollary}

\newpage

\begin{proposition}
\label{thm:4.4}
Yet another formula worked out in [LK] is the following recoupling formulas
\begin{equation}  
\label{eq:4.4}
\includegraphics[scale=.9]{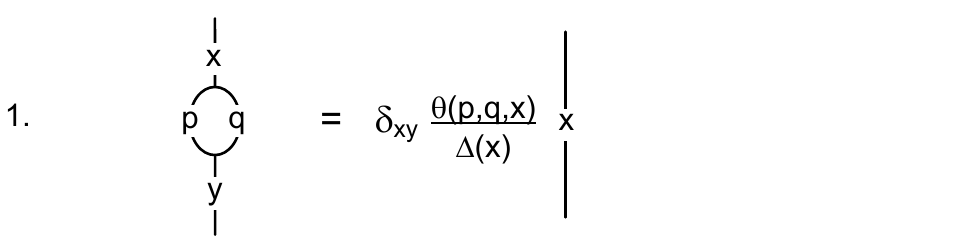}  
\end{equation}  
\begin{equation}  
\label{eq:4.44}
\includegraphics[scale=.9]{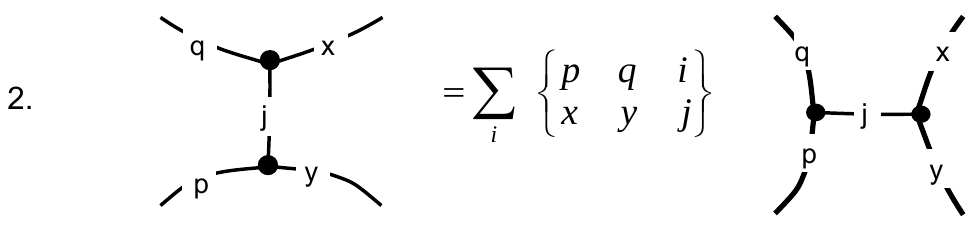}   
\end{equation}  
where the coefficients are known as 6j-symbols and are defined as 
\begin{equation}  
\label{eq:4.6}
\left\{\begin{array}{ccc} {p} & {q} & {i} \\ {x} & {y} & {j} \end{array}\right\}
=\frac{\left[\begin{array}{ccc} {p} & {q} & {i} \\ {x} & {y} & {j} \end{array}\right]\, \Delta _{i} }
         {\theta (p,y,i)\, \theta (q,x,i)} 
\end{equation}  
\end{proposition}

\noindent Let us add to these results the following shortcut

\begin{proposition}
\label{thm:4.5}
A triangular fragment of the network may be simplified to a single vertex as shown: 
\begin{equation}  
\label{eq:kirchoff}
\includegraphics[scale=1]{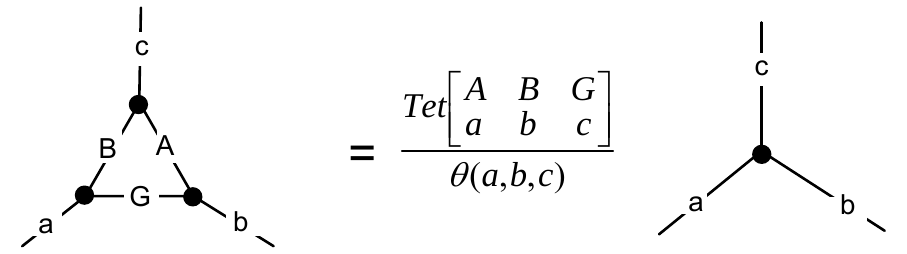}
\end{equation}  
\end{proposition}

\begin{proof}   
Apply Eq. \eqref{eq:4.44} to the lower bar segment labeled $G$, and then \eqref{eq:4.4} to the resulting (sum of) loops with exits.  
(The delta function chooses from the sum only one entry.)  
Use the symmetries in Corollary \ref{thm:4.3} to bring the formula to the simpler form.  
\end{proof}   
\

\noindent 
{\bf Conclusions.} 
The new basis for the spin network formalism proposed in this paper is compared with the standard one in the table below.  
We want to mention these two features of the formalism presented here:

\begin{enumerate}
\item  Very simple in handling and natural.  In particular they are derived naturally and do not require 
imposing artificial rules for sign changes at crossing for fixing inconsistencies. 
This natural simplicity should contribute to popularization of the theory. 

\item  Consistency with the representation theory of ${\rm SL}(2)$ (${\rm SU}(2)$).  
Recall that the spin networks are motivated by the irreducible representations of this group (and Lie algebra).  
They are defined on {\it symmetric} tensor products of 2-dimensional spinor space. 
This, hence the symmetrization tensor used to resolve labeled graphs is more in harmony with the initial motivation 
than the antisymmetrization tensor used in the standard approach.  
\end{enumerate}
$$
\includegraphics[scale=.8]{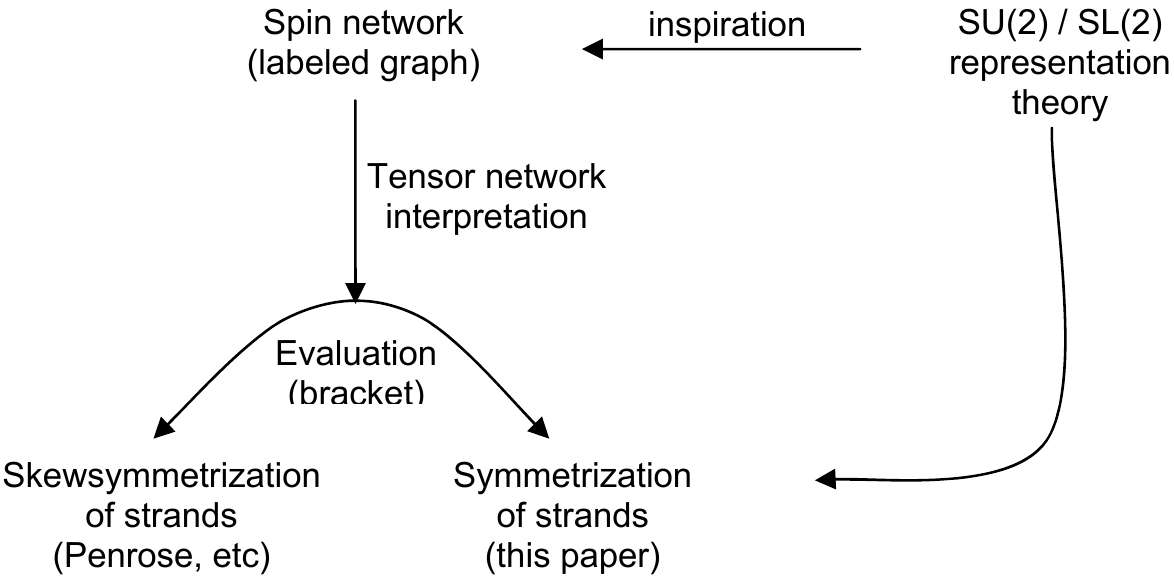}
$$
Remarkably, both versions of rules lead to the same numerical values in chromatic evaluations of closed spin networks. 
Possible further extensions of the formalism in physics might need 
however the symmetric case since antisymmetrization of tensor products of more than 2 vectors vanishes. 
$$
\includegraphics[scale=.84]{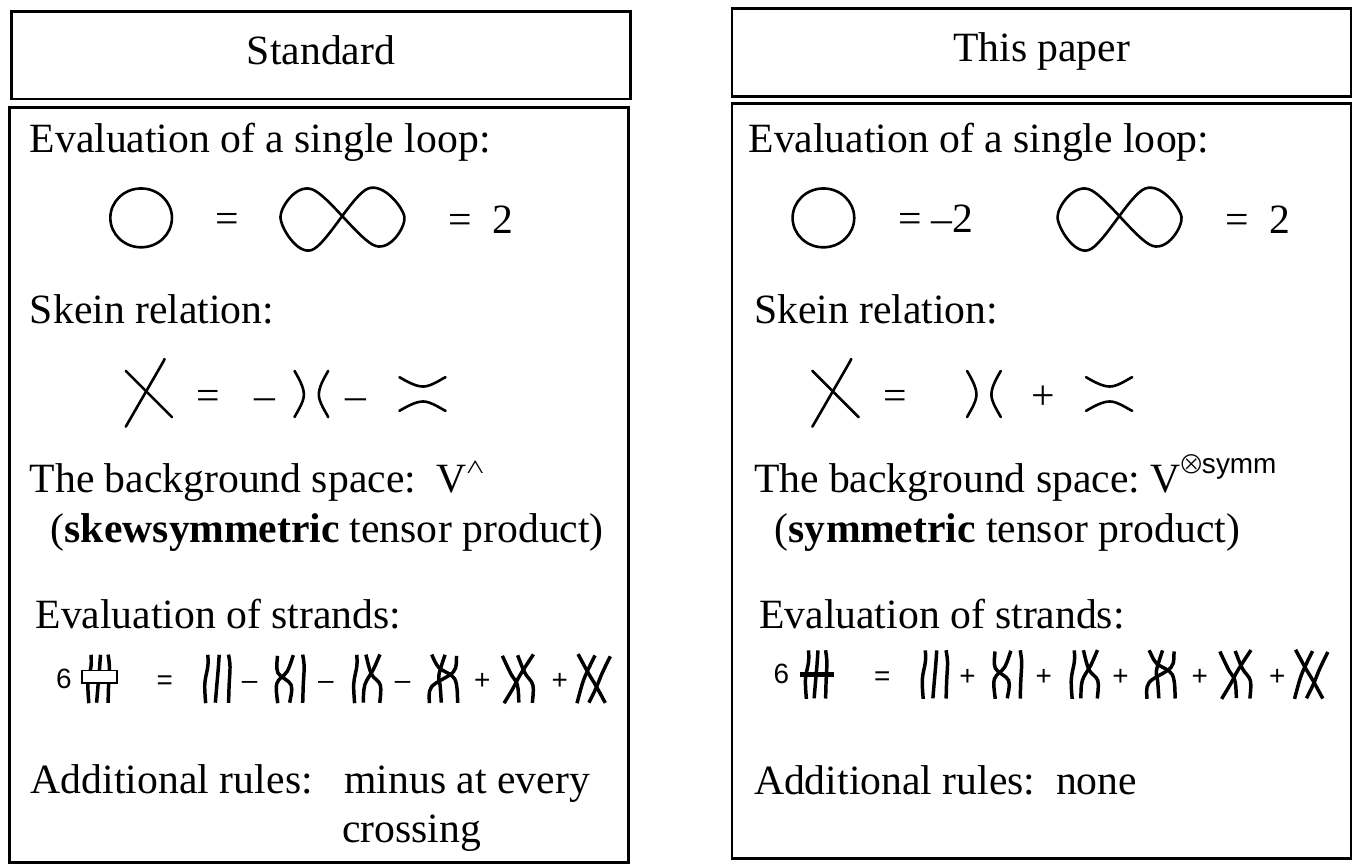}
$$

\newpage

\section{Apollonian disk packing as a spin network}
\label{sec:5}
Apollonian disk packing and other arrangements of mutually tangent circles make a natural source of spin networks.  
This is especially interesting in the light of many features of Apollonian disk packings pointing 
to a possible future test model for loop quantum gravity.  
They contain already such ingredients as Minkowski space-time metric as well as spin structure.  

\smallskip

\noindent 
To turn a configuration of tangent circles into a spin network one needs to:
\vspace{-3pt}

\begin{enumerate}
\item  
Replace every ideal triangle (space between the disks) by a vertex
\item  
Replace every tangency point by an edge and give it the label equal to the sum of curvatures of the two disks.
\end{enumerate}
\vspace{-3pt}

Note that three tangent circles determine two ideal triangles (inner and outer and therefore lead to two vertices:
$$
\includegraphics[scale=1]{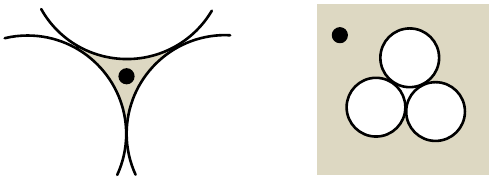}
$$
Here is the idea:  Consider a system of disks of integer curvatures.  
Turn every region into a vertex and every tangency to an edge of a graph.  
The edges are labeled by the sum of the curvatures of the tangent circles.  
This way we obtain a spin network. Here are some simple examples (without labels):
$$
\includegraphics[scale=.9]{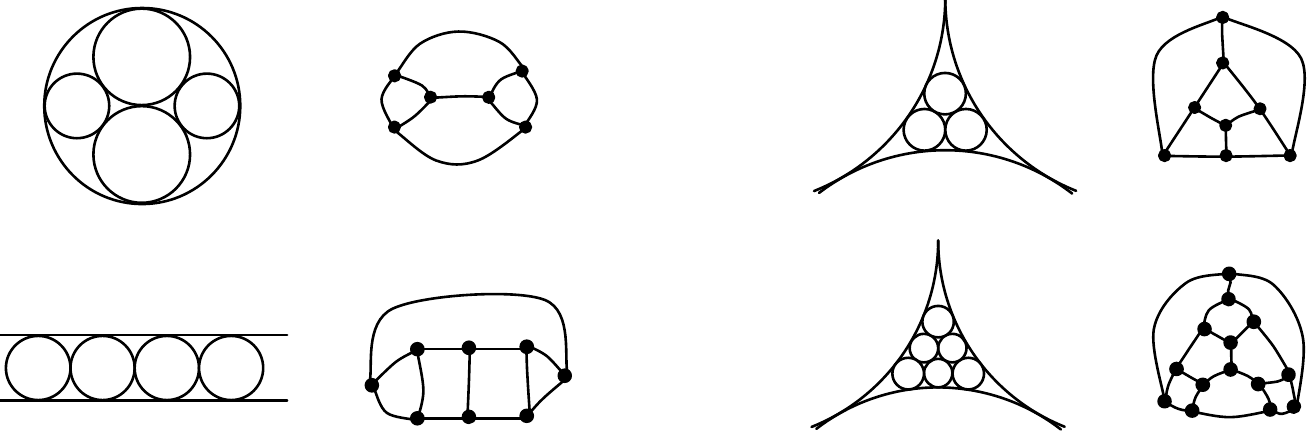}
$$
A more exotic example based on the Ford circle theorem is moved to the end of the section. 

\smallskip
The arrangements of circles lead thus to chromatic evaluations. For instance three mutually tangent circles lead to the theta-type evaluation: 
$$
\includegraphics[scale=.9]{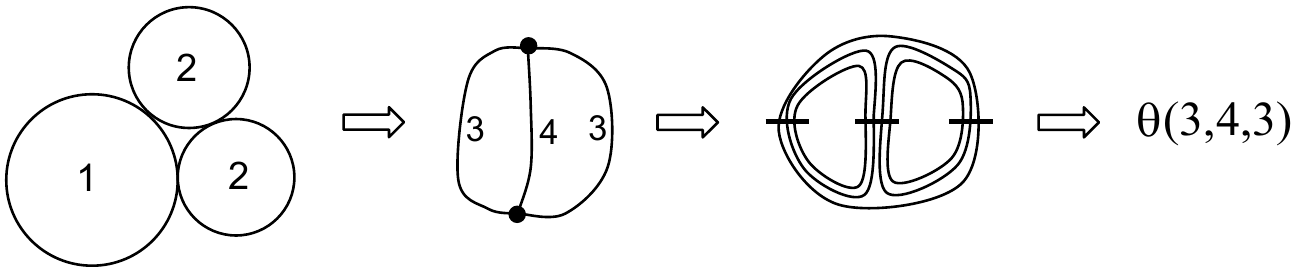}
$$
Here is the list of the three basic situations with small number of mutually tangent disks with curvatures labeled $a,b,c,d$.  
The evaluations $\langle G \rangle$ of the resulting nets will be denoted by bracketed list of the curvatures with an upper index ``o''. 
$$
\includegraphics[scale=.84]{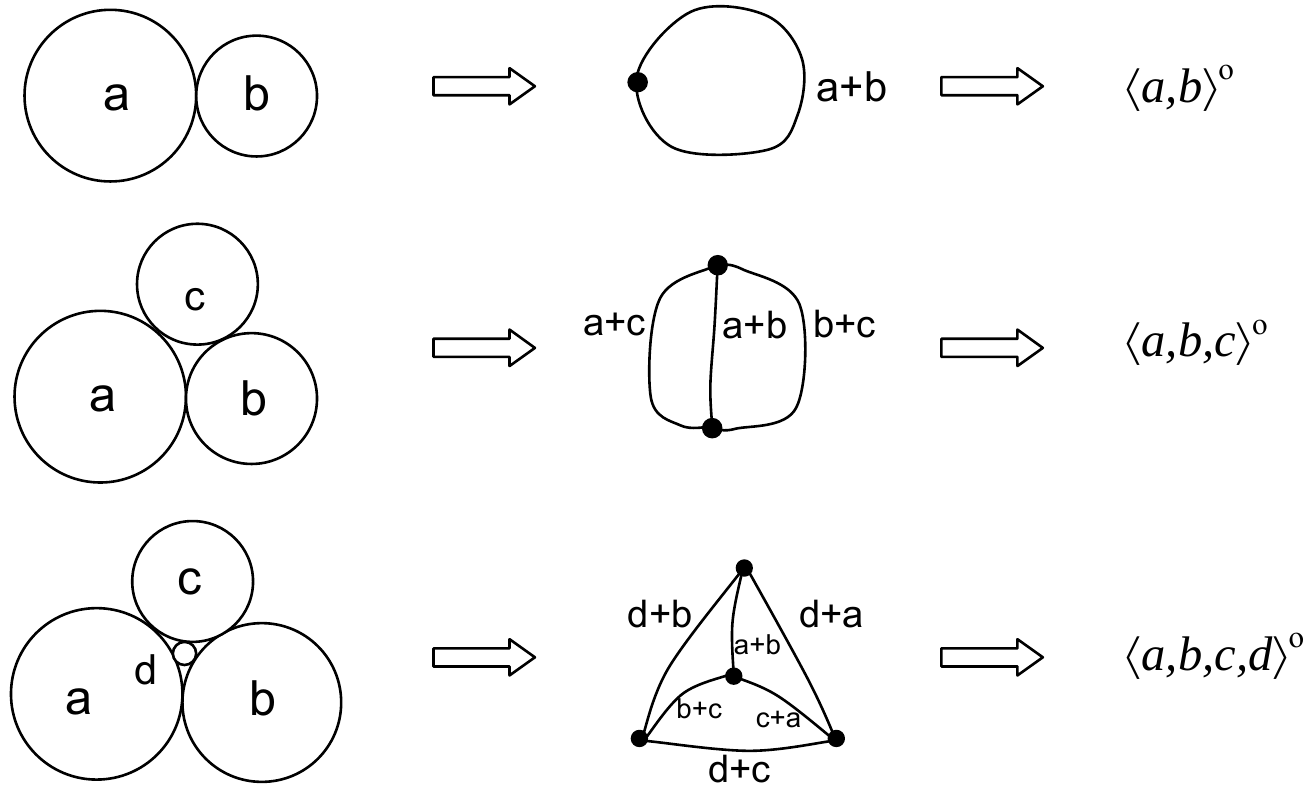}
$$
The last case is called Descartes configuration.

\medskip

\noindent { \bf Notation.}  The following notation is a convenient shortcut but also a way to reveal some symmetries:
\begin{equation}  
\label{eq:5.1}
                              \displaystyle {a+b\choose a,\,b } =
                              \dfrac{(a+b)!}{a!\,b!} 
\,, \qquad
                              \displaystyle {a+b+c\choose a,\,b,\,c } =
                              \dfrac{(a+b+c)!}{a!\,b!\,c!} 
 , \quad etc 
 \end{equation}  
%
\begin{proposition}
\label{thm:5.1}
The chromatic functions for the first three simplest disk arrangements may be expressed in the following way:

\

\noindent A. Two-circle configuration:
\begin{equation}  
\label{eq:5.2}
\left\langle a,b\right\rangle ^{o}  (-1)^{a+b} (a+b+1)
\end{equation}

\noindent B. Three-circle configuration:
\begin{equation}  
\label{eq:5.3}
\left\langle a,b,c\right\rangle ^{o}  
          = \; \dfrac{
                              a!\, b!\, c!\, (a+b+c+1)!
                              }
                              {
                              (a+b)!\, 
                              (b+c)!\, 
                              (c+a)! 
}
 \end{equation}

\noindent C. Four-circle configuration:
\begin{equation}  
\label{eq:5.4}
\left\langle a,b,c,d\right\rangle ^{o}  
          = \sum _{k=0}^{m} (-1)^{S+k} \; \dfrac{
                              \displaystyle {a\choose k} 
                              \displaystyle {b\choose k} 
                              \displaystyle {c\choose k} 
                              \displaystyle {d\choose k} 
                              \displaystyle {S+1\choose a,b,c,d,1} 
                              }
                              {
                              \displaystyle {a+b\choose a} 
                              \displaystyle {a+c\choose a} 
                              \displaystyle {a+d\choose a} 
                              \displaystyle {b+c\choose b} 
                              \displaystyle {b+d\choose b} 
                              \displaystyle {c+d\choose c} 
                              \displaystyle {S+1\choose k} 
}
 \end{equation}  
where $ S = a+b+c+d,  \quad m = {\rm min}\{a,b,c,d\}$.

\end{proposition}

\begin{proof}   
The first two formulas result by the following evaluations of (\ref{eq:4.1}a)  and (\ref{eq:4.1}b) of Proposition \ref{thm:4.1}:
$$  
\left\langle a,b\right\rangle ^{o}  = \Delta (a+b)  \,, \qquad
\left\langle a,b,c\right\rangle ^{o} = \theta (a\!+\!b,\,b\!+\!c,\,c\!+\!a)
$$

\noindent and simplification.  The Descartes configuration is somewhat involved.  We start with 
$$
      Tet\left[\begin{array}{ccc} {P} & {Q} & {R} \\ {p} & {q} & {r} \end{array}\right]
       =
    Tet\left[\begin{array}{ccc} {c\!+\!d}\;\;  & {a\!+\!d}\;\;  & {b\!+\!d} \\ {a\!+\!b}\;\;  & {b\!+\!c}\;\;  & {c\!+\!a} \end{array}\right]
$$
Referring to (\ref{eq:4.1}c), note that all $b_i$ terms are equal: 
$$
                   b_{1} =b_{2} =b_{3} =a+b+c+d\quad \Rightarrow \quad \min \{ b_{i} \} =a+b+c+d\equiv S
$$
where $S$ will denote the sum of all four curvatures, $S=a+b+c+d$.  
As to terms $a_i$, each corresponds to one of these
$$
      a_1 = S - a, \quad  a_2 = S - b, \quad   a_3 = S - c, \quad  a_4 = S - d\,,
$$
therefore
$$
        {\rm max} \{a_i\}   =  S - {\rm min} \{a,b,c,d\}
$$

\noindent 

\noindent The trick is to run the sum backwards.  Define $k = S - s$.  
The term in front of the sum in (\ref{eq:4.1}c)  becomes:

\begin{equation}  
\label{eq:5.5}
\frac{\prod _{i,j}(b_{i} -a_{j} )! }{p\, !q!r!P!Q!R!} \ = \ 
\frac{(a!)^{3} (b!)^{3} (c!)^{3} (d!)^{3} }{(a+b)\, !(b+c)!(c+a)!(d+a)!(d+b)!(d+c)!} 
\end{equation}  
The sum in (\ref{eq:4.1}c) becomes:
\begin{equation}  
\label{eq:5.6}
\sum _{s=\max \{ a_j \} }^{\min \{ b_i \} } \frac{(-1)^{s} (s+1)!}  {\prod _i (s-a_i )!\prod _j (b_j -s)!  } 
\ =\ 
\sum _{k=0}^{m} \frac{(-1)^{S-k} (S-k+1)!}   {(a-k)!(b-k)!(c-k)!(d-k)!\; \cdot \; k!k!k!} 
\end{equation}  
Putting these two terms together we get a ratio of products of various factorials.  
It takes a number of moves to transform it to the form given in the proposition. 
The reader may however easily check that \eqref{eq:5.5} resolves to the product 
of \eqref{eq:5.5} and \eqref{eq:5.6} by expanding all binomials, followed by a number of 
evident cancellations.  
\end{proof}

\noindent 
The evaluation of the Descartes configuration may be abbreviated:

\begin{equation}
\label{eq:5.7}
\left\langle a,b,c,d\right\rangle ^{o}  
             =    \sum _{k=0}^{m} (-1)^{S+k} \; 
\frac{  
           \displaystyle  {S+1 \choose a,b,c,d,1}   \prod _{i}  {x_{i}  \choose k} 
              }{
         \displaystyle {S+1 \choose k} \; \prod _{i<j}    { x_i +x_j  \choose  x_i ,\, x_j } 
               } 
\end{equation}

\noindent 
The above chromatic evaluations follow a certain pattern in progression from 2- to 3- to 4-circle configuration,  
which may be displayed by enforcing the binomial coefficients in the formulas as follows: 

\begin{proposition}[uniform notation]  
\label{thm:5.2}
The may be expressed in the following way

\noindent 
\begin{equation}  \label{eq:5.8}
\left\langle a,b\right\rangle ^{o}  
  =  (-1)^{a+b} \frac{\left(\begin{array}{c} {a+b+1} \\ {a,\; b,\; 1} \end{array} \right)}
                               {\left(\begin{array}{c} {a+b} \\ {a,\; b} \end{array}\right)} 
\end{equation}  
$$ 
\left\langle a,b,c\right\rangle ^{o}  
   =  (-1)^{a+b+c} \frac{\left(\begin{array}{c} {a+b+c+1} \\ {a,\; b,\; c,\; 1} \end{array}\right)}
                        {\left(\begin{array}{c} {a+b} \\ {a,\; b} \end{array}\right)
                          \left(\begin{array}{c} {b+c} \\ {b,\; c} \end{array}\right)
                          \left(\begin{array}{c} {c+a} \\ {c,\; a} \end{array}\right)} 
$$
$$
\left\langle a,b,c,d\right\rangle ^{o}  
          = \sum _{k=0}^{m} (-1)^{S+k} \; \dfrac{
                              \displaystyle {a\choose k} 
                              \displaystyle {b\choose k} 
                              \displaystyle {c\choose k} 
                              \displaystyle {d\choose k} 
                              \displaystyle {a+b+c+d+1\choose a,\,b,\,c,\,d,\,1} 
                              }
                              {
                              \displaystyle {a+b\choose a} 
                              \displaystyle {a+c\choose a} 
                              \displaystyle {a+d\choose a} 
                              \displaystyle {b+c\choose b} 
                              \displaystyle {b+d\choose b} 
                              \displaystyle {c+d\choose c} 
                              \displaystyle {S+1\choose k} 
}
$$
$ S = a+b+c+d, \quad  m = \min\{a,b,c,d\}$
\end{proposition}

\noindent 
Now the pattern of progression becomes evident.

\

\begin{corollary}  
The ratios of the consecutive evaluations are: 
\begin{equation}  \label{eq:5.9}
\frac{\left\langle a,b,c\right\rangle ^{o} }{\left\langle a,b\right\rangle ^{o} } 
= 
(-1)^c\;
\frac{       \displaystyle {1+a+b+c\choose c} 
                  }{
                              \displaystyle {a+c\choose c} 
                              \displaystyle {b+c\choose c} 
                   }
\end{equation}

$$
\frac{  \left\langle a,b,c,d\right\rangle ^{o} }
       {  \left\langle a,b,c\right\rangle ^{o} } 
 =  
\sum _{k=0}^{m} (-1)^{d+k} \; \frac{
                              \displaystyle {a\choose k} 
                              \displaystyle {b\choose k} 
                              \displaystyle {c\choose k} 
                               \displaystyle {1+a+b+c+d-k\choose d-k} 
                  }{
                              \displaystyle {a+d\choose d} 
                              \displaystyle {b+d\choose d} 
                              \displaystyle {c+d\choose d} }
$$  
\end{corollary}

\begin{corollary}
\label{thm:5.4}
The inclusion of a circle of curvature $d$ in the ideal triangle determined by circles 
of curvature $a,b,c$ in a disk packing changes its chromatic evaluation by the factor defined by
the second equation of \eqref{eq:5.9}.
\end{corollary}

The future work in chromatic evaluations of disk configurations will include the asymptotic behavior of integral disk packing
in the sense of recursive accumulation of disks in a particular Apollonian disk packing. Potential application 
in the budding quantum gravity is one of the motivations. 

\newpage
\noindent 
{\bf Example:}  
Ford disk arrangement \cite{Ford} is an arrangement of disks tangent to the real axis with the famous property of ``counting'' the rational numbers.  
It consists of discs tangent to the real line, drawn at every $p/q \in \mathbb Q$  with radius $1/q^2$. 
If two discs at $p/q$ and $p'/q'$ are tangent, there is a third one generated in the ideal triangle formed by them and the real axis. 
It is tangent at the fraction that is the Farey sum:
\begin{equation}  \label{eq:}
\frac{p}{q} \oplus \frac{p'}{q'} =\frac{p+p'}{q+q'} 
\end{equation}  

~

\includegraphics[scale=.9]{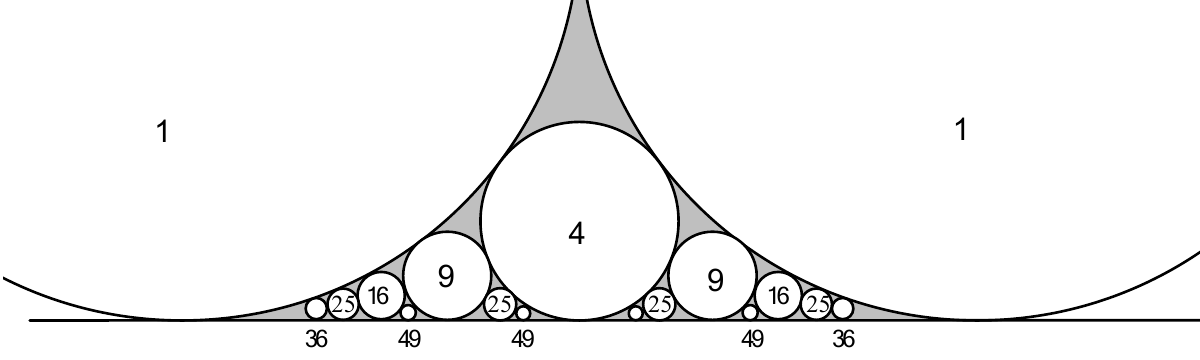}

\

\noindent 
Here are the resulting weights of the initial stages of building the Ford arrangement: 
$$
\includegraphics[scale=.9]{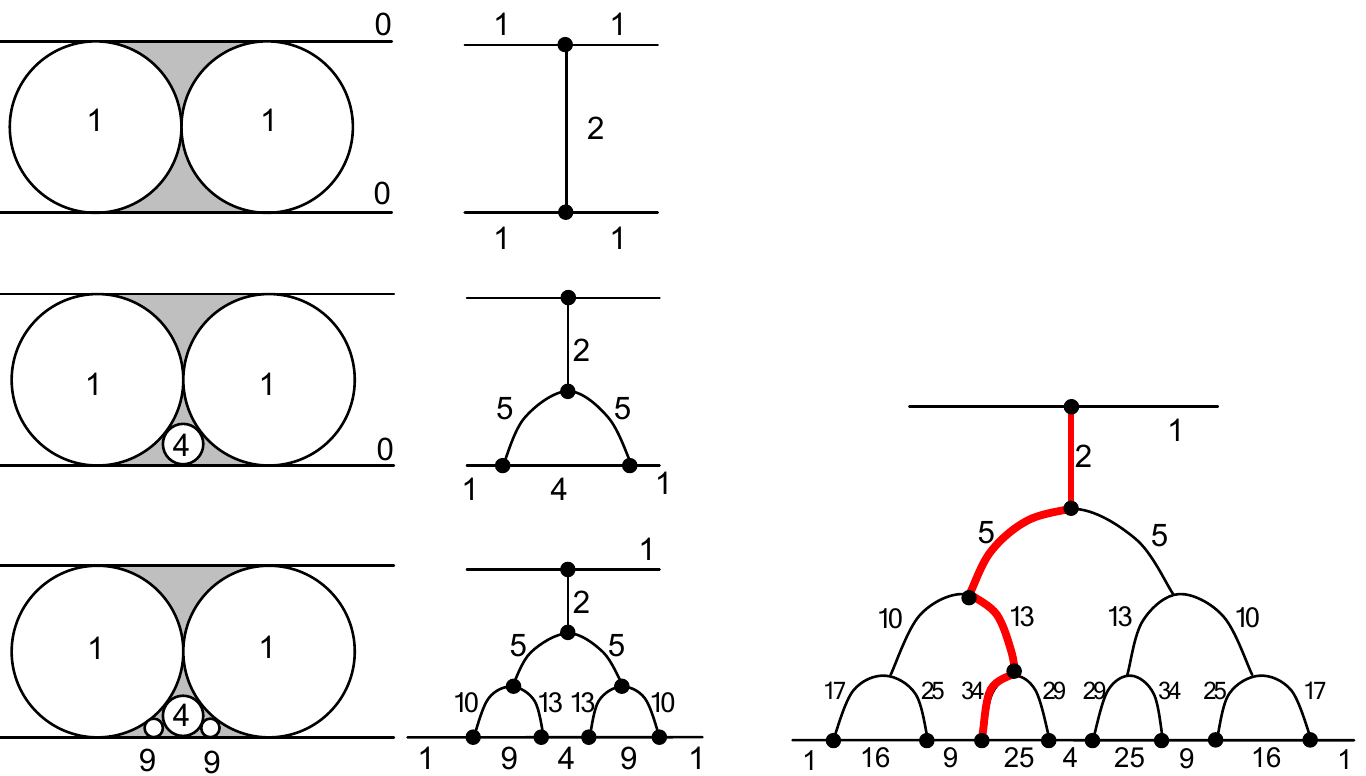}
$$
The spin networks have many number-theoretic features.  For starter, the dark zigzag is made of every other Fibonacci number (bold)
$$
            \mathbf 1, \ 1,\  \mathbf 2,\ 3,\  \mathbf 5,\ 8,\  \mathbf {13},\ 21,\  \mathbf {34},\ 55,\  \mathbf {89},\ \dots 
$$
The bottom labels are squares.  The edges along the sides, (2,5,10,17,\dots .) are squares increased by 1.

\subsection{Chromatic evaluations: From combinatorial to analytic}

One may find somewhat baffling puzzling and disturbing the fact that the chromatic values of nets jump between opposite signs 
when a label increases by a unit.    
For instance, the delta function for two tangent circles is $\Delta(a,b) = (-1)^{a+b}(a+b+1)$.
Setting $a=1$ and varying the other circle, we get:
$$
\Delta(1,1) = +3,\quad
\Delta(1,2) = -4,\quad
\Delta(1,3) = +5,\quad
...  etc.
$$
One could hope for a more ``steady'' behavior.

In the case of circles, it makes sense to consider circles of non-integer curvatures.
And, luckily, this brings the solution to the conundrum.
The arguments of chromatic functions \eqref{eq:5.2}--\eqref{eq:5.4}
have natural extensions beyond $\mathbb Z$
but the values are in the field of complex numbers.
Thus, e.g. $\Delta (1,1)=+3$ moves to $\Delta(1,2)=-4$
along a path in the Argand plane, 
see Figure \ref{Argand}.

\begin{figure}
\centering
\includegraphics[scale=.5]{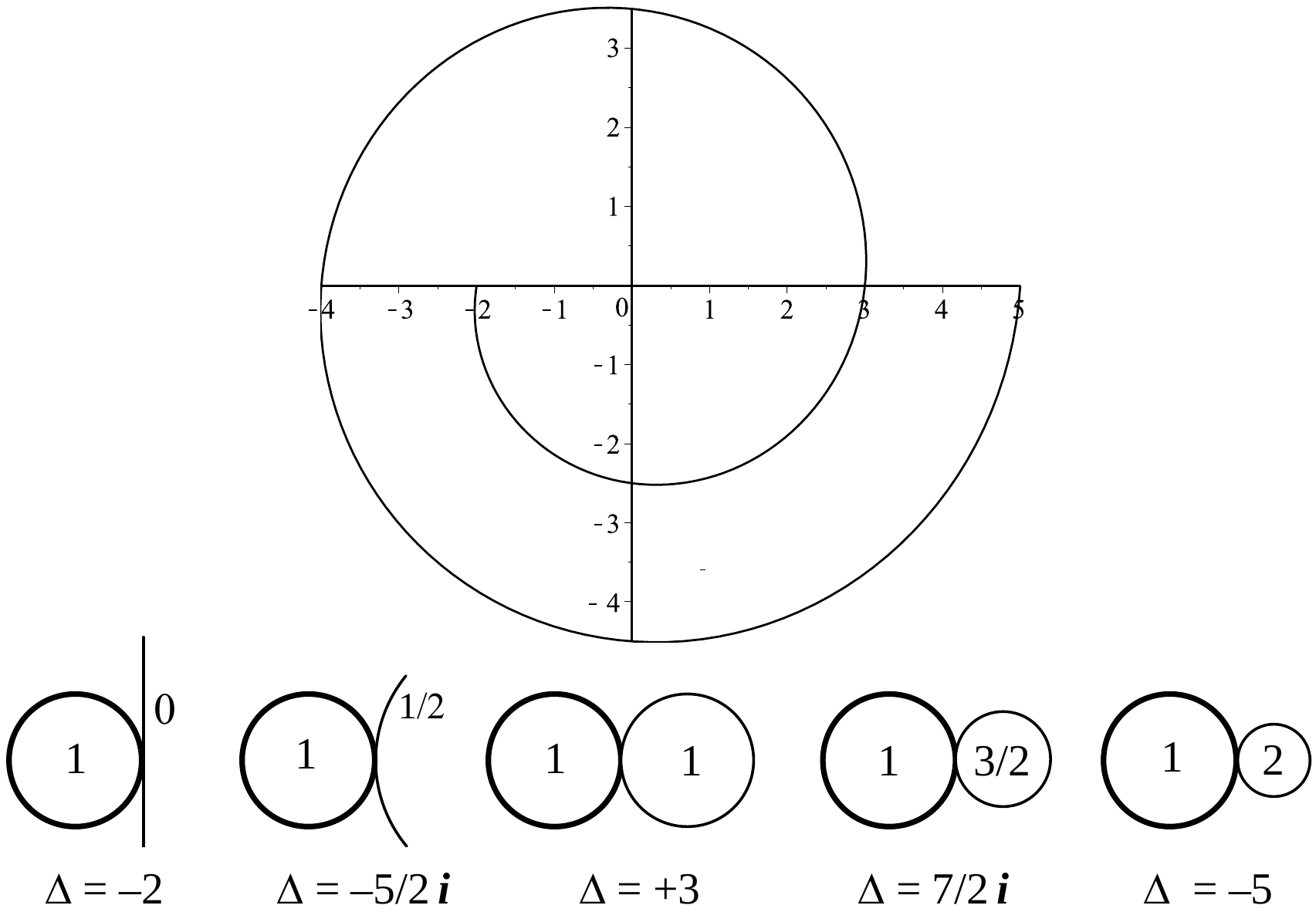}
\caption{The chromatic functions assume complex numbers for non-integer curvatures}
\label{Argand}
\end{figure}

%

The situation is explained by the fact that the delta function may be written as 
a complex valued function on a real plane:
$$
\Delta(x,y) = (1+x+y)\, {\rm e}^{(x+y)\pi i}
$$
The spiral in Figure \ref{Argand} is thus 
$\Delta(1,x) = (2+x)\, {\rm e}^{(1+x)\pi i}$.

Quite interestingly, the other two functions, since composed of binomials, also admit such extensions.
This we have a situation: 
$$
\begin{array}{rll}
\hbox{combinatorics:} & \mathbb Z^n &\quad \longrightarrow \quad \mathbb Q \\
\hbox{analysis:}         & \mathbb R^n &\quad \longrightarrow \quad \mathbb C
\end{array}
$$

The behavior of these functions may be found intriguing.
For instance, consider the three-circle configurations:
$$
\theta: \mathbb R^3 \  \longrightarrow \ \mathbb C
$$
To visualize this function, we look for the image of lines.
They may differ considerably as the directions of the lines change.
The image a straight path $(1,1,x)$ in $\mathbb R^3$
is an ever-growing outward spiral (Figure \ref{fig:three}, right).
The image of the straight line $(1,x,x)$ (two circles changing size simultaneously)
is bounded: it is a spiral that approaches a circle as its attractor from inside
(Figure \ref{fig:three}, center).
The image of the diagonal line $(x,x,x)$ (three circles changing size simultaneously)
is an inward spiral towards 0
(Figure \ref{fig:three}, left).
See Figure \ref{fig:interesting} for other examples.

\begin{figure}[h]
$$
\begin{array}{ccc}
\includegraphics[scale=.37]{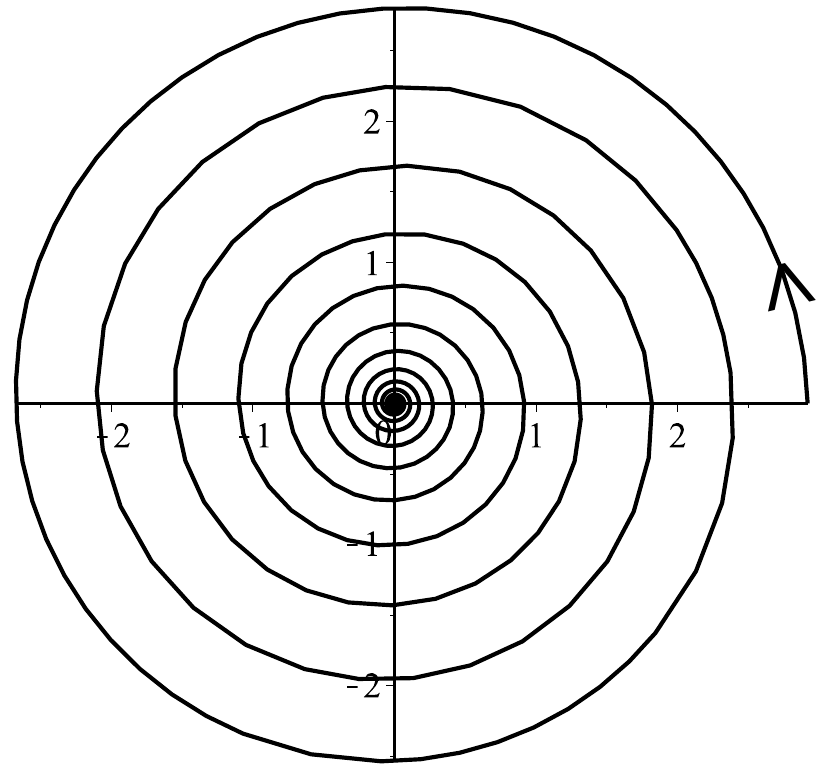}
&\qquad
\includegraphics[scale=.35]{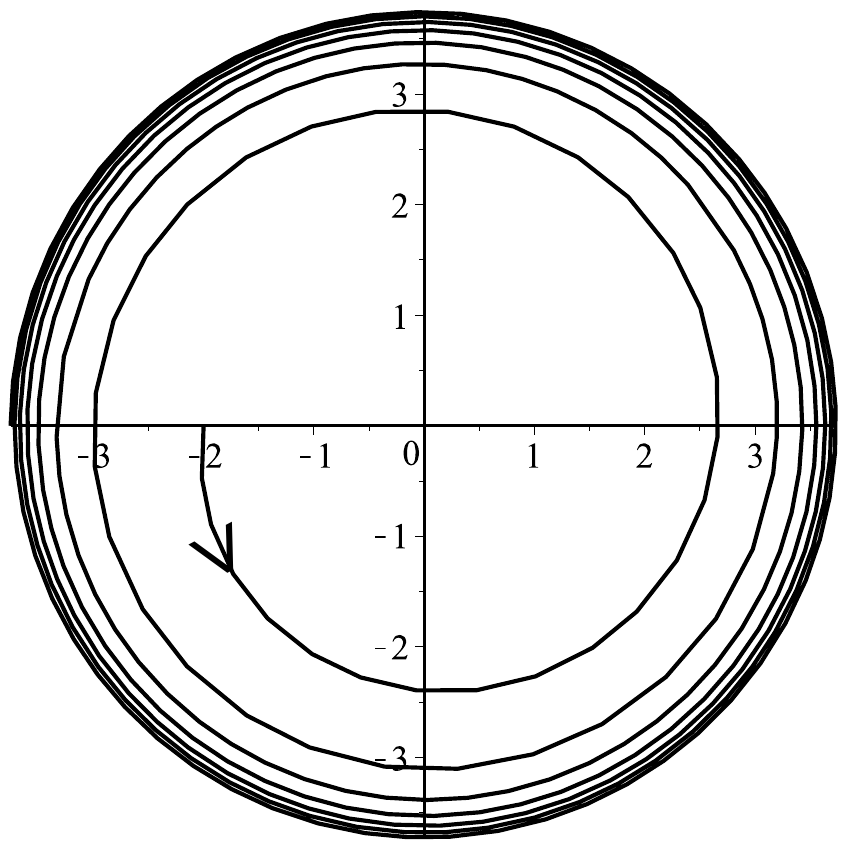}
&\includegraphics[scale=.35]{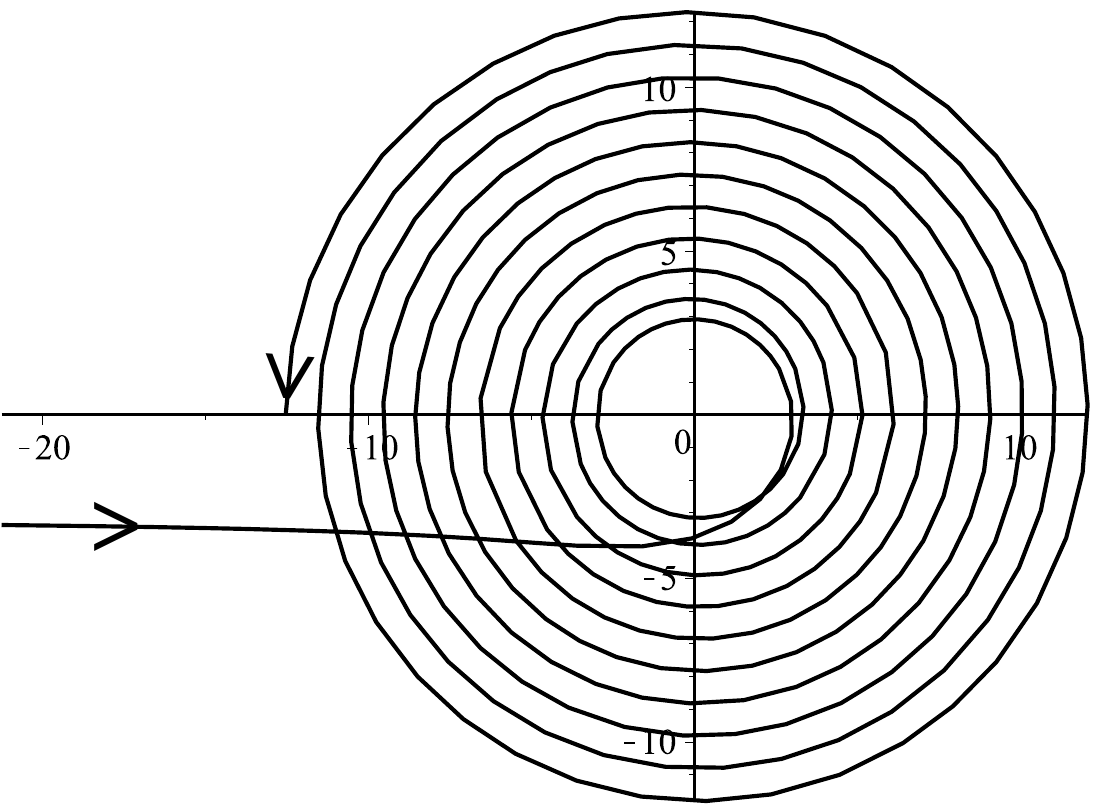}\\
\theta(x,x,x)
&\qquad
\theta(1,x,x)
&\qquad
\theta(1,1,x)
\\
\hbox{\small Sink}
&\qquad
\hbox{\small Atractor}
&\qquad
\hbox{\small Steady growth}
\end{array}
$$
\caption{Three behaviors of $\theta$ along a line.}
\label{fig:three}
\end{figure}

These complex-valued functions and their potential significance
for spin networks and quantum gravity remain to be studied further.

\begin{figure}[h]
$$
\begin{array}{ccc}
\includegraphics[scale=.5]{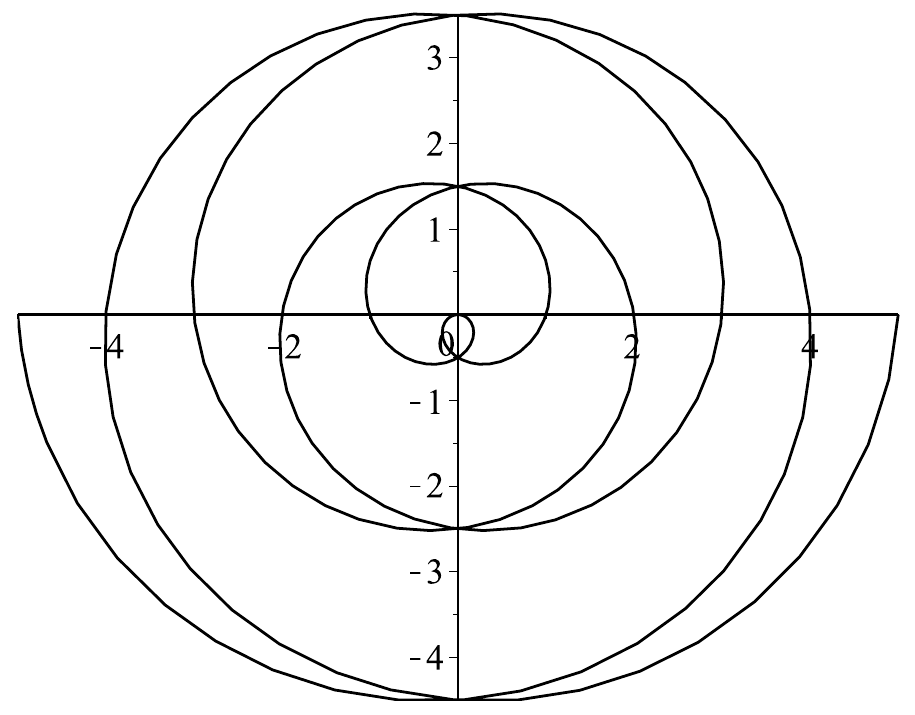}
&\qquad
\includegraphics[scale=.5]{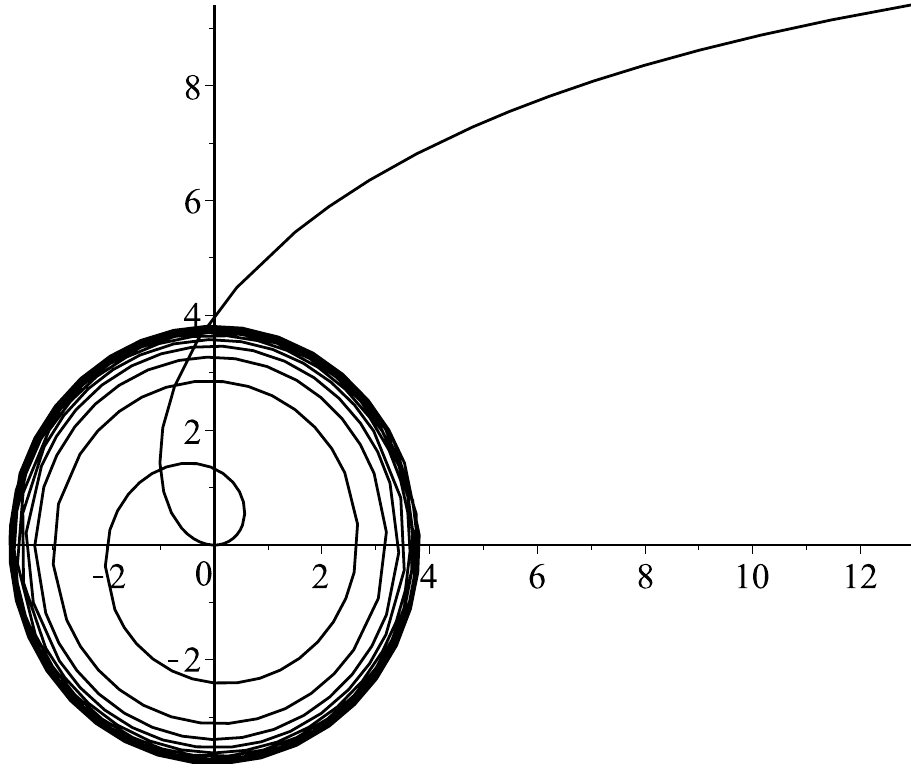}\\
\Delta(1,x)
&\qquad
\theta(1,x,x)\\
x=-3...7
&\qquad
x=-.9... 11
\end{array}
$$
\caption{Examples extended}
\label{fig:interesting}
\end{figure}


\newpage

\newpage


\begin{thebibliography}{00}

\bibitem{ba}
J.H. Barrett,  Skein spaces and spin structures, {\it Math. Proc. Camb. Phil. Soc.} {\bf 126} (1999)  267.

\bibitem{BS} D.M. Brink and G.R. Satchler,
{\it Angular Momentum}, 2nd edn, Oxford University Press, 1968).

%


\bibitem{Ford} 
Laster R. Ford, 
Fractions,
{\it The American Mathematical Monthly}  {\bf 45} (1938) 2586--601.


\bibitem{Ka} Louis H. Kauffman,
{\it  Knots and Physics (Knots and Everything)}, (World Scientific, 3 edition, 2001).

\bibitem{KLo} 
Louis H. Kauffman and Samuel J. Lomonaco, 
Spin Networks and Quantum Computation, 
{\it Bulg. J. Phys.}  {\bf 35} (2008) 241--256.

\bibitem{KL}  Louis H. Kauffman and S\'ostenes L. Lins, 
{\it Temperley-Lieb recoupling theory and invariants of 3-manifolds},
(Annals of Mathematics Studies, vol. 134, Princeton University Press, Princeton, 1994).


\bibitem{Lev}   I. B. Levinson, 
Sums of Wigner coefficients and their graphical representation,
{\it Proceed. Physical-Technical Inst. Acad. Sci. Lithuanian SSR},  {\bf 2} (1956) 17-30.


\bibitem{Lev1} 	J.N. Levinson 
Tr. Fiz.-Tekh. Inst., Ashkhabad, {\bf 2}, (1957) 31.  

\bibitem{Lev2} 	J.N. Levinson {\it Liet. TSR Mosklu Akad.  Darb.}  {\bf B4} (1957)  3.

\bibitem{Ma} 	Seth A. Major, A Spin Network Primer, {\it Am.J.Phys.}  {\bf 67} (1999) 972-980.


\bibitem{Pe} 	Roger Penrose,  Angular momentum: an approach to combinatorial space-time,
in {\it Quantum Theory and Beyond}, ed. T. Batin (Cambridge University Press, 1971), pp. 151--180.

\bibitem{Ro} 	Carlo Rovelli, {\it Quantum Gravity}, (Cambridge University Press, 2004). 

\bibitem{YB} 	A.P. Yutsis and A.A. Bandziaitis, {\it Quantum Theory of Angular Momentum} (Mintus, Vilnius, 1965).

\bibitem{YLV} 	A.P. Yutsis, J.B. Levinson, V.V. Vanagas, 
Mathematical Apparatus of the Theory of Angular Momentum, 
{\it Israeli Program for Scientific Translation}  (Jerusalem, 1962).

\end{thebibliography}
\end{document}